%% file: arxiv_version.tex
\documentclass[runningheads]{llncs}

\usepackage{graphicx}

\usepackage[dvipsnames]{xcolor}
\usepackage{mathtools}
\usepackage{tabularray}
\usepackage{comment}
\UseTblrLibrary{booktabs}
\usepackage{makecell}
\usepackage{arydshln}
\usepackage{svg}
\usepackage{orcidlink}
\usepackage{marvosym}
\input{macros}

\definecolor{causecolor}{RGB}{215, 95, 76}
\definecolor{effectcolor}{RGB}{142, 196, 222}

\usepackage{hyperref}
\usepackage[firstpage]{draftwatermark}

\SetWatermarkAngle{0}

\SetWatermarkText{\ifodd\thepage 
 \raisebox{20cm}{%
    \hspace{-6.7cm}%
    \href{https://doi.org/10.5281/zenodo.10946309}{\includegraphics{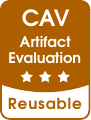}}}
 \else
    \raisebox{20cm}{%
\hspace{11.45cm}%
\href{https://doi.org/10.5281/zenodo.10946309}{\includegraphics{3-reusable}}}
 \fi
}

\begin{document}

\title{Synthesis of Temporal Causality}

\author{Bernd Finkbeiner\,\orcidlink{0000-0002-4280-8441} \and Hadar Frenkel\,\orcidlink{0000-0002-3566-0338}\and Niklas Metzger\,\orcidlink{0000-0003-3184-6335} \and Julian Siber\textsuperscript{(\Letter)}\orcidlink{0000-0003-0842-0029}}

\authorrunning{Finkbeiner et al.}

\institute{CISPA Helmholtz Center for Information Security, Saarbrücken, Germany\\ \email{\{finkbeiner,hadar.frenkel,niklas.metzger,julian.siber\}@cispa.de}}

\maketitle 

\begin{abstract}
We present an automata-based algorithm to synthesize $\omega$-regular causes for $\omega$-regular effects on executions of a reactive system, such as counterexamples uncovered by a model checker. 
Our theory is a generalization of \emph{temporal causality}, which has recently been proposed as a framework for drawing causal relationships between trace properties on a given trace.
So far, algorithms exist only for verifying a single causal relationship and, as an extension, cause synthesis through enumeration, which is complete only for a small fragment of effect properties.
This work presents the first complete cause-synthesis algorithm for the class of $\omega$-regular effects. We show that in this case, causes are guaranteed to be $\omega$-regular themselves and can be computed as, e.g., nondeterministic Büchi automata. We demonstrate the practical feasibility of this algorithm with a prototype tool and evaluate its performance for cause synthesis and cause checking.

\keywords{Actual causality \and Cause synthesis \and Reactive systems \and Temporal logic \and Büchi automata}
\end{abstract}

\section{Introduction}

Causality is a key ingredient for explaining model-checking results~\cite{BallNR03,ChakiGS04,Leitner-FischerL13,SallingerWZ23} and a reasoning tool in several verification and synthesis algorithms~\cite{BaierCFFJS21,KupriyanovF13,KupriyanovF14}. These techniques have retrofitted causality definitions from philosophy~\cite{Hume1748,Lewis73b} and artificial intelligence~\cite{HalpernP01}, which were not designed for reactive systems with infinite dynamics and often fall short in such ad-hoc applications. For instance, 
popular approaches for explaining model-checking results highlight the counterexample trace at events that constitute causes~\cite{BeerBCOT09,CoenenDFFHHMS22,HorakCMHFMDFD22}. Yet, marking a (possibly infinite) set of events does not clearly describe 
the temporal behavior manifested by them since, e.g., two events can be individually responsible for the effect or only together. Similarly, the occurrence of events in the loop part of a trace can be relevant, e.g., only once or infinitely often.

To address such reoccurring problems arising with causal reasoning in reactive systems, Coenen et al.\ have recently proposed \emph{temporal causality} for drawing causal relationships between temporal properties on a given trace of a system~\cite{CoenenFFHMS22}. Causal properties can then be described symbolically with logics or automata, which give a concise description of the possibly infinite causal behavior, and are, moreover, amenable to verification algorithms.

\subsection{Temporal Causality}

At its core, temporal causality uses counterfactual reasoning to infer a causal relationship: A property is a cause for some effect property on a given trace, where both properties hold, if on all closest traces that do not satisfy the cause, the effect is not satisfied either. Additionally, the cause property has to be semantically minimal. Hence, it is a form of \emph{actual causation}~\cite{Halpern16}, which describes the concrete causal behavior in the given, \emph{actual} observation (the trace), and not all of the system behavior that may cause the effect (which loosely corresponds to the concept of \emph{general causation}).\footnote{Actual and general causality are also called token and type causality in the literature.}

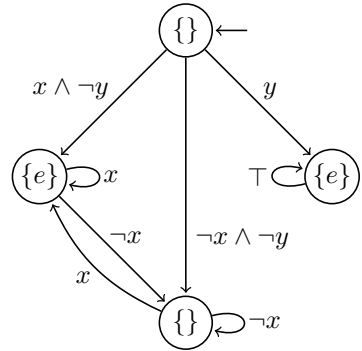
\begin{wrapfigure}{R}{0.45\linewidth}
	\centering\vspace{-2em}
         \begin{tikzpicture}[->,shorten >=1pt,semithick,auto,node distance=2.75cm, on grid,initial text=,
			every state/.style={minimum size=27pt,inner sep=1pt}]
         \node[state,initial right](1){$\{\}$};
         \node[state,below left = of 1](2){$\{e\}$};
         \node[state,below right = of 1](3){$\{e\}$};
         \node[state,below right = of 2](4){$\{\}$};
         \path[->,draw](1) edge[] node[swap,xshift=0.2em,yshift=-0.2em]{$x \land \lnot y$} (2)
         (2) edge[] node[pos=0.45,xshift=-0.3em,yshift=-0.3em]{$\lnot x$} (4)
         (4) edge[bend left=20] node[xshift=0.2em,yshift=0.2em]{$x$} (2);
         \path[->,draw](1) edge[] node[xshift=-0.2em,yshift=-0.2em]{$y$} (3)
         (3) edge[loop left] node[xshift=0.2em]{$\top$} (3);
         \path[->,draw](2) edge[loop right] node[xshift=-0.2em]{$x$} (2);
         \path[->,draw](1) edge[] node[pos=0.765]{$\lnot x \land \lnot y$} (4)
         (4) edge[loop right] node[xshift=-0.2em]{$\lnot x$} (4);
         \end{tikzpicture}
         \caption{Example system.}\label{fig:example1}
         \vspace{-1.5em}
\end{wrapfigure}

To illustrate, consider the system depicted in Fig.~\ref{fig:example1}, where $x$ and $y$ are inputs and $e$ is an output. We are interested in what input behavior causes the effect $\LTLeventually e$ on the trace $\pi = (\{x,e\})^\omega$ -- we skip the output label of the first position. Our first guess may be $y \lor \LTLeventually x$, which characterizes all system traces that satisfy $\LTLeventually e$. However, this is too general to describe the causal behavior on $\pi$. After all, 
the left disjunct~$y$ is not even satisfied by $\pi$. 
Let us see which condition fails. The counterfactual criterion holds: The closest system traces that do not satisfy $y \lor \LTLeventually x$ also do not satisfy the effect, as these are exactly the traces that go directly to the lower state labeled with the empty set and loop there infinitely. However, minimality is not satisfied, as the property $\LTLeventually x$ implies $y \lor \LTLeventually x$ (i.e., is semantically smaller) and also satisfies the counterfactual criterion: the closest trace that does not satisfy it is $(\{\})^\omega$. In particular, the existence of, e.g., trace $\{y,e\}(\{\})^\omega$ that also does not satisfy the cause $\LTLeventually x$, but still satisfies the effect $\LTLeventually e$, is irrelevant, as $(\{\})^\omega$ is closer to $\pi$ than the trace $\{y,e\}(\{\})^\omega$. It is worth pointing out that we only measure distance over inputs. Picking a property that is \emph{too} small fails the counterfactual criterion: If we picked $\LTLglobally x$, which implies $\LTLeventually x$, there would be, e.g., the closest trace $\{\}(\{x,e\})^\omega$ that still satisfies the effect.

In their original work~\cite{CoenenFFHMS22}, Coenen et al.\ showed that the requirements for a valid causal relationship can be encoded as a hyperproperty~\cite{ClarksonS10}, such that checking whether a given $\omega$-regular property is indeed the cause for a given $\omega$-regular effect on a trace can be decided via model checking. This has recently been implemented in a sketch-based algorithm for enumerating causes~\cite{BeutnerFFS23}, which is complete for effects containing $\LTLnext$ as the only temporal operator. That approach, of course, covers only a tiny fragment of the original theory. How to compute the cause for an arbitrary $\omega$-regular effect has remained an open question.

\subsection{Contributions and Structure}

As it turns out, the intricate balance between the counterfactual criterion and minimality of temporal causality gives rise to an intuitive order-theoretic characterization of causes: The complement of the cause is the upward closure of the negated effect property in the partial order defined by the similarity relation (measuring distance from the actual trace). We illustrate the intuition behind this characterization in Section~\ref{subsec:topology}, and formally prove it in Section~\ref{subsec:characterization}.

The consequence of our characterization is that if we can compute the upward closure of the negated effect $\overline{\Effect}$ and the complement of the result, then we can compute the cause for $\Effect$ on $\pi$. We show that if $\Effect$ is an $\omega$-regular property, $\pi$ in a lasso shape, and the similarity relation is also defined by a (relational) $\omega$-regular property, such an upward closure can be constructed as a nondeterministic Büchi automaton, which means that the cause (i.e., the complement of the automaton) again is an $\omega$-regular property. This approach forms the core of our cause synthesis algorithm, which we describe in Section~\ref{sec:synthesis}.

The complexity of our algorithm significantly scales in the size of the description of the similarity relation, which is problematic due to the complex and large similarity relations of previous work. Coenen et al.~\cite{CoenenFFHMS22} observed that with the original counterfactual criterion, these similarity relations need to satisfy the assumption that there is a non-empty set of closest traces for any actual trace and candidate cause, otherwise the counterfactual condition can be vacuously true. We tie this restriction to the \emph{limit assumption} first introduced by Lewis~\cite{Lewis73a} and study similarity relations through this lens. Concrete similarity relations that have been proposed so far~\cite{BeutnerFFS23,CoenenFFHMS22} satisfy the limit assumption by adding additional criteria, but these increase the size of the formula describing the similarity relation significantly. In Section~\ref{sec:definitions}, we show that we can instead modify the counterfactual condition of the causality definition to allow similarity relations that do not satisfy the limit assumption, using Lewis' semantics for counterfactuals~\cite{Lewis73a}, as extended to non-total similarity relations by Finkbeiner and Siber~\cite{FinkbeinerS23}. Crucially, this modification retains the original semantics of Coenen et al.\ for similarity relations that satisfy the limit assumption as long as the actual trace is deterministic. Hence, it generalizes our closure-based characterization and the corresponding algorithm to significantly simpler similarity relations.

In Section~\ref{sec:experiments}, we show through experiments with our prototype tool CORP that our modified counterfactual criterion leads to significantly faster computations in practice. We further compare our cause synthesis algorithm with the incomplete sketching approach of the tool CATS~\cite{BeutnerFFS23}. Last, we extend our approach to cause checking through cause synthesis with an additional equivalence check, which we compare with the checker implemented in CATS.

\paragraph{Contributions.} In summary, we make the following contributions:
\begin{itemize}
    \item We extend the theory of temporal causality to similarity relations that do not satisfy the limit assumption.
    \item We prove an order-theoretic characterization of causes as downward closed sets of the similarity relation.
    \item Based on this characterization, we develop the first complete method for $\omega$-regular cause synthesis.
    \item We present and evaluate a prototype implementation of our approach.
\end{itemize}

\section{Preliminaries}

We start by recalling preliminaries regarding our system model. Then, we provide background on automata and logics for describing temporal properties.

\subsubsection*{Systems and Traces.}
We model \emph{systems} as nondeterministic finite state machines $\mathcal{T} = (S, s_0, \AP, \delta, l)$ where 
$S$ is a finite set of \emph{states}, $s_0 \in S$ is the \emph{initial state}, $\AP = \In \cupdot \Out$ is the set of \emph{atomic propositions} consisting of \emph{inputs} $\In$ and \emph{outputs} $\Out$, $\delta: S \times 2^I \rightarrow 2^S$ is the \emph{transition function} determining a set of successor states for a given state and input, and $l: S \rightarrow 2^\Out$ is the \emph{labeling function} mapping each state to a set of outputs. 
A \emph{trace} of $\mathcal{T}$ is an infinite sequence $\pi = \pi[0] \pi[1] \ldots \in (2^\AP)^\omega$, with $\pi[i] = A \cup l(s_{i+1})$ for some $A \subseteq I$ and $s_{i+1} \in \delta(s_i,A)$ for all $i \geq 0$, i.e., we skip the label of the initial state in the first position.
$\mathit{traces}(\mathcal{T})$ is the set of all traces of $\mathcal{T}$.  For two subsets of atomic propositions $\Vap,\Wap \subseteq \AP$, let $\Vap|_\Wap = \Vap \cap \Wap$, $\pi|_\Wap = \pi_0|_\Wap\,\pi_1|_\Wap\ldots$ and $\pi =_\Vap \pi'$ iff $\pi|_\Vap = \pi'|_\Vap$ for traces $\pi,\pi'$. A trace $\pi_0$ is \emph{deterministic} in $\mathcal{T}$ iff for all $\pi_1 \in \mathit{traces}(\mathcal{T}): \pi_0 =_I \pi_1 \rightarrow \pi_0 = \pi_1$. A trace $\pi$ is \emph{lasso-shaped}, if there exist $i,j=i+1,k \in \mathbb{N}$ such that $\pi = \pi_0 \ldots \pi_i \cdot (\pi_{j} \ldots \pi_k)^\omega$, we then define $| \pi | = k - 1$.

\subsubsection*{Büchi Automata.} A \emph{nondeterministic B\"uchi automaton} (NBA)~\cite{Buechi62Decision} is a tuple $\mathcal{A} = (Q,\Sigma, Q^0, F, \Delta)$, where $Q$ denotes a finite set of \emph{states}, $\Sigma$ is a finite \emph{alphabet}, $Q^0 \subseteq Q$ is a set of \emph{initial states}, $F\subseteq Q$ is the set of \emph{accepting states}, and $\Delta: Q \times \Sigma \rightarrow 2^Q$ is the \emph{transition function} that maps a state and a letter to a set of possible successor states. The \emph{size} of an NBA $|\mathcal{A}|$ is the number of its states $|Q|$.
A \emph{run} of~$\mathcal A$ on an infinite \emph{word} $w = w_1w_2 \ldots \in \Sigma^{\omega}$ is an infinite sequence $ r = q_0q_1\ldots \in Q^{\omega}$ with $q_0 \in Q^0$ and $q_{i+1} \in \Delta(q_i,w_i)$ for all $i \in \mathbb N$. 
A run $r$ of the NBA is \emph{accepting} if there exist infinitely many $i \in \mathbb{N}$ such that $q_i \in F$. 
The \emph{language} $\Lang(\mathcal{A})$ is the set of all words that have an  accepting run. We say that some trace property  $\prop \subseteq (2^A)^\omega$ is \emph{$\omega$-regular}, if there is an NBA  $\mathcal{A}$ such that $\Lang(\mathcal{A}) = \prop$. A trace $\pi$ \emph{satisfies} any $\prop \subseteq (2^A)^\omega$, denoted by $\pi \models \prop$, iff $\pi|_A \in \prop$.

\subsubsection*{Linear-time Temporal Logic.} We use \emph{Linear-time Temporal Logic} (LTL)~\cite{Pnueli77} to succinctly specify a fragment of $\omega$-regular properties throughout the paper. LTL formulas are built using the following grammar, where $a \in \AP$:
\begin{equation*}
\varphi \Coloneqq a \mid \neg \varphi \mid \varphi \land \varphi \mid \LTLnext \varphi \mid \varphi \LTLu \varphi \enspace .
\end{equation*}
The semantics of LTL are given by the following satisfaction relation, which recurses over the positions $i$ of the trace $\pi$.
\begin{equation*}
\begin{array}{lll}
\pi,i \models a       & \text{iff } & a \in \pi[i] \\
\pi,i  \models \neg \varphi              & \text{iff } & \pi,i  \nmodels \varphi \\
\pi,i  \models \varphi \land \psi         & \text{iff } & \pi,i  \models \varphi \text{ and } \pi,i \models \psi \\
\pi,i \models \X \varphi                & \text{iff } & \pi,i+1 \models \varphi \\
\pi,i \models \varphi\U\psi             & \text{iff } & \exists j \geq i \text{ such that } \pi,j \models \psi \text{ and } \forall i \leq k < j \ldot \pi,k \models \varphi
\end{array}
\end{equation*}
A trace $\pi$ \emph{satisfies} a formula $\varphi$, denoted by $\pi \models \varphi$  iff the formula holds at the first position: $\pi,0 \models \varphi$. 
The \emph{language} $\Lang(\varphi)$ is the set of all traces that satisfy a formula $\varphi$. We also consider the usual derived Boolean connectives: $\lor$, $\rightarrow$, $\leftrightarrow$; and temporal operators: $\varphi \R \psi \equiv \neg(\neg \varphi \U \neg \psi)$, $\F \varphi \equiv \true \U \varphi$, $\G \varphi \equiv \false \R \varphi$.

\subsubsection*{Relational Properties.} 
Relational properties, or, \emph{hyperproperties}~\cite{ClarksonS10}, allow us to relate multiple system executions, and reason about their interaction. Counterfactual causality is a hyperproperty, and in particular, temporal causality as defined by Coenen et al.\ was shown to be a hyperproperty~\cite{CoenenFFHMS22}. 
Many logics to express temporal hyperproperties have been suggested in recent years (e.g., \cite{GutsfeldMO20,BeutnerF23LMCS,BaumeisterCBFS21,BeutnerFFM23}), the most prominent one being HyperLTL~\cite{ClarksonFKMRS14}. In this paper, we do not use a hyperlogic to express temporal causality, but we use the related notion of \emph{zipped traces} (e.g., \cite{BeutnerF23c}) for definig similarity relations. 
A \emph{zipped trace} of three traces $\pi_{0,1,2}$ is defined as $\mathit{zip}(\pi_0,\pi_1,\pi_2)[i] = \{(a,t_k) \; | \;  a \in \pi_k[i]\}$, i.e., we construct the zipped trace from disjoint unions of the positions of the three traces, where atomic propositions from the traces $\pi_{0,1,2}$ are distinguished through pairing them with the trace variables $t_{0,1,2}$.

\section{Overview: The Topology of Causality}

Our main results on cause synthesis heavily rely on a characterization of causes as certain downward closed sets of system traces that are ordered by a similarity relation. We illustrate the main intuition behind this characterization in Section~\ref{subsec:topology}. Then, in Section~\ref{subsec:limit}, we outline how we extend this result to more general similarity relations than originally considered by Coenen et al.~\cite{CoenenFFHMS22}. 

\subsection{Actual Causes as Downward Closed Sets of Traces}\label{subsec:topology}

Our central theorem states that the temporal cause for an effect $\Effect$ on some actual trace $\pi$ is the largest subset of $\Effect$ that is downward closed\footnote{$X \subseteq \mathit{traces}(\mathcal{T})$ is downward (upward) closed in $(\mathit{traces}(\mathcal{T}),\leq_\pi)$ if for all $\pi_x \in X$ and $\pi_t \in \mathit{traces}(\mathcal{T})$, $\pi_t \leq_\pi \pi_x$ ($\pi_x \leq_\pi \pi_t$) implies $\pi_t \in X$.} in the preordered set of system traces $(\mathit{traces}(\mathcal{T}),\leq_\pi)$, where $\leq_\pi$ is a (comparative) similarity relation that orders traces based on their similarity to $\pi$. Figure~\ref{fig:proset} illustrates this abstractly. Arrows together with nodes represent system executions, whose traces form $\mathit{traces}(\mathcal{T})$ and are ordered by the irreflexive reduction $<_\pi$ of the similarity relation. The set of system traces is, in general, infinite, such that there may be infinitely many other traces which are omitted from the illustration for sake of clarity. However, note that similarity relations must be designed such that all traces are further away from the actual trace $\pi$ than itself, i.e., $\pi$ is a minimum of $\leq_\pi$. The set of traces that satisfy the effect is depicted by the area that is colored in light blue. The actual trace $\pi$ is an element of this set, as this is the trace on which the cause for a given effect is analyzed. 

Coenen et al.'s temporal causality is counterfactual in nature, and now requires that the \emph{closest} traces outside of the cause $\Cause$, which in Figure~\ref{fig:proset} is marked by the red border, do not satisfy the effect. In the illustration, this is reflected by $\pi_b$ and $\pi_c$ not satisfying the effect, i.e., not being in a light blue area. At the same time, Coenen et al.\ require the cause to be the smallest set that satisfies this, which means that only traces that satisfy the effect are included: Otherwise, the upward closure\footnote{The upward closure of a set $X$ is the smallest upward closed set containing $X$.} of traces that do not satisfy the effect could be removed. Hence, in Figure~\ref{fig:proset} the area inside the red border is light blue. 

\begin{figure}[t] 
     \centering  
         \begin{subfigure}{.5\textwidth} 
         \centering
            \begin{tikzpicture}[>=latex,semithick,label distance =5pt,node distance = 0.75cm]
                \draw[rounded corners=1mm,fill=effectcolor,draw=causecolor,line width=2.3pt] (0.0,0.0) -- (0.5,-2.5)  -- (-0.5,-2.5) -- cycle;
                \draw[rounded corners=1mm,color=effectcolor,fill=effectcolor,line width=2.3pt] (0.0,0.0) -- (1.75,-2.5)  -- (2.95,-2.5) -- cycle;
                \draw[rounded corners=1mm,color=effectcolor,fill=effectcolor,line width=2.3pt] (0.0,0.0) -- (-1.75,-2.5)  -- (-2.95,-2.5) -- cycle;
                \node[draw, circle, inner sep=2pt, fill=white] at (0,0)(start) {};
                \node[draw, circle, inner sep=2pt, fill=white, below =0.7 of start](a2) {};
                \node[draw, circle, inner sep=2pt, fill=white, below =0.7 of a2](a1) {};
                \node[below =0.375 of a1](a0) {};
                \node[draw, circle, inner sep=2pt, fill=white, below left=0.763 and 0.25 of start](b2) {};
                \node[draw, circle, inner sep=2pt, fill=white, below left=0.763 and 0.25 of b2](b1) {};
                \node[below left=0.3815 and 0.065 of b1](b0) {};
                \node[draw, circle, inner sep=2pt, fill=white, below left=0.763 and 0.7 of start](c2) {};
                \node[draw, circle, inner sep=2pt, fill=white, below left=0.763 and 0.7 of c2](c1) {};
                \node[below left=0.3815 and 0.315 of c1](c0) {};
                
                \node[draw, circle, inner sep=2pt, fill=white, below right=0.763 and 0.25 of start](d2) {};
                \node[draw, circle, inner sep=2pt, fill=white, below right=0.763 and 0.25 of d2](d1) {};
                \node[below right=0.3815 and 0.065 of d1](d0) {};
                \node[draw, circle, inner sep=2pt, fill=white, below right=0.763 and 0.7 of start](e2) {};
                \node[draw, circle, inner sep=2pt, fill=white, below right=0.763 and 0.7 of e2](e1) {};
                \node[below right=0.3815 and 0.315 of e1](e0) {};
                
                \draw[-stealth,decorate,decoration={snake,segment length=3.5pt,amplitude=0.5pt,pre length=1pt,post length=3pt}] (start) -- (a2);
                \draw[-stealth,decorate,decoration={snake,segment length=3.5pt,amplitude=0.5pt,pre length=1pt,post length=3pt}] (a2) -- (a1);
                \draw[-stealth,decorate,decoration={snake,segment length=3.5pt,amplitude=0.5pt,pre length=1pt,post length=3pt}] (a1) -- (a0);
                \draw[-stealth,decorate,decoration={snake,segment length=3.5pt,amplitude=0.5pt,pre length=1pt,post length=3pt}] (start) -- (b2);
                \draw[-stealth,decorate,decoration={snake,segment length=3.5pt,amplitude=0.5pt,pre length=1pt,post length=3pt}] (b2) -- (b1);
                \draw[-stealth,decorate,decoration={snake,segment length=3.5pt,amplitude=0.5pt,pre length=1pt,post length=3pt}] (b1) -- (b0);
                \draw[-stealth,decorate,decoration={snake,segment length=3.5pt,amplitude=0.5pt,pre length=1pt,post length=3pt}] (start) -- (c2);
                \draw[-stealth,decorate,decoration={snake,segment length=3.5pt,amplitude=0.5pt,pre length=1pt,post length=3pt}] (c2) -- (c1);
                \draw[-stealth,decorate,decoration={snake,segment length=3.5pt,amplitude=0.5pt,pre length=1pt,post length=3pt}] (c1) -- (c0);
                
                \draw[-stealth,decorate,decoration={snake,segment length=3.5pt,amplitude=0.5pt,pre length=1pt,post length=3pt}] (start) -- (d2);
                \draw[-stealth,decorate,decoration={snake,segment length=3.5pt,amplitude=0.5pt,pre length=1pt,post length=3pt}] (d2) -- (d1);
                \draw[-stealth,decorate,decoration={snake,segment length=3.5pt,amplitude=0.5pt,pre length=1pt,post length=3pt}] (d1) -- (d0);
                \draw[-stealth,decorate,decoration={snake,segment length=3.5pt,amplitude=0.5pt,pre length=1pt,post length=3pt}] (start) -- (e2);
                \draw[-stealth,decorate,decoration={snake,segment length=3.5pt,amplitude=0.5pt,pre length=1pt,post length=3pt}] (e2) -- (e1);
                \draw[-stealth,decorate,decoration={snake,segment length=3.5pt,amplitude=0.5pt,pre length=1pt,post length=3pt}] (e1) -- (e0);
            
                \node[below = 0 of a0](p){$\boldsymbol{\pi}$};
                \node[below = 0 of c0](pa){$\boldsymbol{\pi_a}$};
                \node[below = 0 of b0](pb){$\boldsymbol{\pi_b}$};
                \node[below = 0 of d0](pc){$\boldsymbol{\pi_c}$};
                \node[below = 0 of e0](pd){$\boldsymbol{\pi_d}$};
                \node[left = -0.1 of p]{$>_\pi$};
                \node[right = -0.1 of p]{$<_\pi$};
                \node[left = -0.1 of pb]{$>_\pi$};
                \node[right = -0.1 of pc]{$<_\pi$};
                \node[draw=effectcolor,fill=effectcolor,rectangle,rounded corners,line width=2.3pt,inner sep=5pt]at(2.75,0){$\Effect$};
                \node[draw=causecolor,ultra thick,rectangle,rounded corners,line width=2.3pt,inner sep=5pt]at(2.75,-0.85){$\Cause$};
            \end{tikzpicture}
         \caption{The cause $\Cause$ as a subset of the effect $\Effect$.}\label{fig:proset}
     \end{subfigure}%
     \begin{subfigure}{.5\textwidth} 
        \centering\hspace{-2em}
            \begin{tikzpicture}[>=latex,thick,label distance=-1pt]
            
                \draw[rounded corners,fill=effectcolor,line width=2.3pt,draw=causecolor] (-0.5,0.7) -- (-0.5,0.3) -- (0.5,0.3) -- (0.5,0.7)  --  cycle;
                
                \draw[rounded corners,fill=effectcolor,line width=2.3pt,draw=effectcolor] (3.0,0.7) -- (3.0,0.3) -- (3.4,0.3) -- (3.4,0.7)  --  cycle;
                \node[] at (1.25,1.35){\textit{Without limit assumption:}};
            
                \node[](pi) at (0,0.5){$\boldsymbol{\pi} \ldots$};
                \node[](c) at (2.25,0.5){$\ldots \leq_\pi \boldsymbol{\pi_i} \ldots \leq_\pi \boldsymbol{\pi_j} \ldots $};
                
                \node[above left= -0.1 and -1.5 of c]{$\forall\exists$};
                \node[above left= -0.1 and -2.95 of c]{$\forall\exists$};
                
                \node[below left= -0.15 and -0.6 of c]{$\infty$};
                
            \end{tikzpicture}\vspace{0.25em}
            \begin{tikzpicture}[>=latex,thick,label distance=-1pt]
            
                \draw[rounded corners,fill=effectcolor,line width=2.3pt,draw=causecolor] (-0.5,0.7) -- (-0.5,0.3) -- (0.5,0.3) -- (0.5,0.7)  --  cycle;
                
                \draw[rounded corners,fill=effectcolor,line width=2.3pt,draw=effectcolor] (2.5,0.7) -- (2.5,0.3) -- (2.9,0.3) -- (2.9,0.7)  --  cycle;
                \node[] at (1,1.35){\textit{With limit assumption:}};
            
                \node[](pi) at (0,0.5){$\boldsymbol{\pi} \ldots$};
                
                \node[](c) at (1.97,0.5){$\leq_\pi \boldsymbol{\pi_i} \ldots \leq_\pi \boldsymbol{\pi_j} \ldots $};
                
                \node[above left= -0.1 and -0.97 of c]{$\forall$};
                
            \end{tikzpicture}\vspace{0.45em}
         \caption{Possible situations at the limit of $\Cause$.}\label{fig:limitassumption}
     \end{subfigure} 
     \caption{Two highlighted aspects of the cause $\Cause$ in the preordered set ($\mathit{traces}(\mathcal{T}),\leq_\pi$). Figure~\ref{fig:proset} illustrates that the cause is the largest downward-closed subset of the effect $\Effect$.
      The quantifiers in Figure~\ref{fig:limitassumption} show which traces outside of the cause are required to avoid the effect in our formalization (\emph{without limit assumption}) and in Coenen et al.'s definition~\cite{CoenenFFHMS22} (\emph{with limit assumption}).} 
\end{figure}
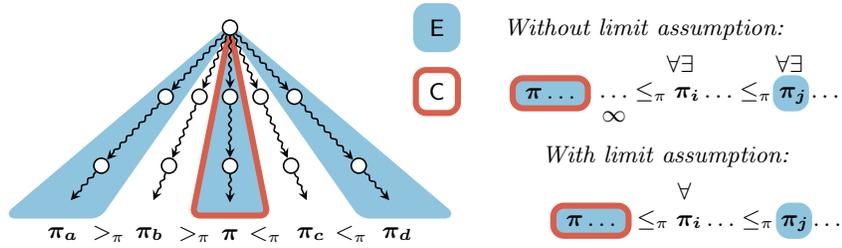

In this paper, we show that the balance between these criteria defines causes that are the largest subsets of $\Effect$ that are downward closed in the preordered set $(\mathit{traces}(\mathcal{T}),\leq_\pi)$. We also propose an algorithm that constructs these causes for effects that are $\omega$-regular properties and traces that are in a lasso-shape. Our algorithm first constructs a nondeterministic Büchi automaton for the complement of the cause $\Cause$. This complement is the upward closure of the negated effect $\overline{\Effect}$, which means it includes all traces for which there exists an at-least-as close trace that does not satisfy the effect. Since $\leq_\pi$ is reflexive, this naturally includes all traces in $\overline{\Effect}$, such as $\pi_b$ and $\pi_c$ in Figure~\ref{fig:proset}. It also includes all traces that are further away than a trace in $\overline{\Effect}$, such as $\pi_a$ and $\pi_d$. 

In the end, these mechanisms make temporal causality a form of \emph{actual} causality that describes a local generalization of the behavior that causes the effect on the actual trace. In the introductory example from Figure~\ref{fig:example1} with the actual trace $\pi = \{x,e\}^\omega$, traces in, e.g., $\Lang(y \land \lnot \LTLeventually x)$ are all further away from $\pi$ than the trace $\{\}^\omega$, which is in $\overline{\Effect} = \Lang(\LTLglobally\lnot e)$. Hence, $\Lang(y \land \lnot \LTLeventually x)$ is included in the upward closure of $\overline{\Effect}$, and none of its elements is included in the cause.

\subsection{Causality Without the Limit Assumption}\label{subsec:limit}
With our approach based on set closure, we can solve a central issue of temporal causality: Since the preordered set $(\mathit{traces}(\mathcal{T}),\leq_\pi)$ is infinite, there only exist traces in $\overline{\Cause}$ that are the \emph{closest} with respect to the actual trace $\pi$ if $(\overline{\Cause},\leq_\pi)$ is well-founded. If this is the case for all possible pairs of actual trace $\pi$ and cause candidate $\Cause$, we say the similarity relation satisfies the \emph{limit assumption}, after Lewis~\cite{Lewis73a}, who formalized it for counterfactual modal logic. Since Coenen et al.'s definition~\cite{CoenenFFHMS22} requires that all \emph{closest} traces avoid the effect, it is restricted to similarity relations that satisfy this assumption. Their counterfactual condition is illustrated in the lower part of Figure~\ref{fig:limitassumption}. If the limit assumption holds, any descending chain $\pi_j \geq_\pi \pi_{j-1} \geq_\pi \ldots$ stabilizes at some $\pi_i$, for which Coenen et al. require $\pi_i \in \overline{\Effect}$. 

If the limit assumption does not hold (upper part of Figure~\ref{fig:limitassumption}), there may be infinite chains $\pi_j\geq_\pi \pi_{j-1} \geq \ldots$ 
 for which a closest $\pi_i$ does not exist. In these instances, Coenen et al.'s criterion would be vacuously true. This is particularly problematic as the canonical similarity relation $\leq^\mathit{subset}_\pi$ does not satisfy the limit assumption. This metric orders two traces as $\pi_{j} \leq^\mathit{subset}_\pi \pi_k$ if the changes between $\pi_j$ and $\pi$ are a subset of the changes between $\pi_k$ and $\pi$. This may lead, for example, to the infinite chain $\{\}^\omega \geq_{\pi}^\mathit{subset} \{x\}\{\}^\omega \geq_{\pi}^\mathit{subset} \ldots$ in the preordered set $(\Lang(\LTLeventually \LTLglobally \lnot x), \leq_{\pi}^\mathit{subset})$, where $\pi = \{x\}^\omega$. Coenen et al.\ add additional constraints on top of $\leq^\mathit{subset}_\pi$ to ensure that it satisfies the limit assumption. These, however, make cause checking more expensive, as observed by Beutner et al.~\cite{BeutnerFFS23}, who therefore combine $\leq^\mathit{subset}_\pi$ with a vacuity check. While this is computationally better, this check simply fails in instances as outlined above, and so certain causes cannot be checked by this method~\cite{BeutnerFFS23}.

 In this work, we solve this conundrum by modifying the definition of temporal causality to accommodate similarity relations that satisfy the limit assumption. We change the central counterfactual condition from a universal quantification over the closest traces in $\overline{\Cause}$ to an $\forall\exists$-quantification over all traces $ \pi_j \in \overline{\Cause}$. 
 For each such trace $\pi_j$, we require the existence of 
  a closer trace $\pi_i \leq_\pi \pi_j$ that does not satisfy the effect. This is depicted in the upper part of Figure~\ref{fig:limitassumption}. Naturally, this quantification mirrors exactly the characterization of cause-complements via upward closed sets (cf.\ Section~\ref{subsec:characterization}). On the theoretical side, we show that if the similarity relation satisfies the limit assumption and a minor assumption on nondeterminism is met, our definition is 
  equivalent to Coenen et al.'s original definition (Section~\ref{subsec:generalization}). On the practical side, we confirm experimentally that our approach leads to significant improvements through the accommodation of simpler similarity relations that do not satisfy the limit assumption (Section~\ref{sec:experiments}).

\section{Generalized Temporal Causality}\label{sec:definitions}

In this section, we generalize the definition of temporal causality to 
accommodate similarity relations that do not satisfy the limit assumption.
We first recall similarity relations and formalize the limit aussmption (Section~\ref{subsec:distance}). Then we present our updated definition of temporal causality (Section~\ref{subsec:definition}). Last, we prove that it retains the original semantics in the special case considered by Coenen et al.\ with a minor additional assumption on nondeterminism (Section~\ref{subsec:generalization}).

\subsection{Similarity Relations and the Limit Assumption}\label{subsec:distance}
A \emph{comparative similarity relation} $\leq_\pi \, \subseteq (2^I)^\omega \times (2^I)^\omega$ is a partial order that 
orders traces by their $\emph{comparative}$ distance from the given actual trace $\pi$, i.e., it gives no quantitative but a relative measurement of distance: $\pi_0 \leq_\pi \pi_1$ means $\pi_0$ is \emph{at-least-as close} to $\pi$ as $\pi_1$. 
We measure distance over the set of inputs $I$, i.e., for two traces $\pi_{0,1} \in (2^\AP)^\omega$ we are only interested in 
$\pi_0|_I \leq_\pi \pi_1|_I$. 
 
If $I$ is clear from the context, we write $\pi_0 \leq_\pi \pi_1$. We require the actual trace to be closer to itself than any other trace, i.e., $\pi \leq_\pi \pi'$ for all $\pi' \in (2^\AP)^\omega$. The ternary relation $\leq$, where $(\pi_0,\pi_1,\pi_2) \in \ \leq$ iff $\pi_1 \leq_{\pi_0} \pi_2$,  
encodes the comparative similarity relations of all possible actual traces $\pi_0$.

\begin{example}\label{ex:subsetdis}
    To illustrate our formalism for similarity relations, consider the following \emph{subset-based} similarity relation $\leq^{\mathit{subset}}$ defined via the zipped trace $\mathit{zip}(\pi_0,\pi_1,\pi_2) \in (2^{\AP \times \{t_0,t_1,t_2\}})^\omega$. To ease comprehension, for some $a \in \AP$ we write $a_\mathit{actual}$ for $(a,t_0)$, $a_\mathit{close}$ for $(a,t_1)$, and  $a_\mathit{far}$ for $(a,t_2)$ to explicitly identify, e.g., propositions on the actual trace, in a given formula. We then have $\pi_{\mathit{close}} \leq^{\mathit{subset}}_{\pi_{\mathit{actual}}} \pi_{\mathit{close}}$ iff
    \begin{align*}
    \mathit{zip}(\pi_{\mathit{actual}},\pi_{\mathit{close}},\pi_{\mathit{far}}) \models \G \bigwedge_{i\in I} \big((i_\mathit{actual} \not\leftrightarrow i_\mathit{close}) \rightarrow (i_\mathit{actual} \not\leftrightarrow i_\mathit{far})\big) \enspace .
    \end{align*}
    For the three traces $\pi_\mathit{actual}, \pi_\mathit{close}, \pi_\mathit{far}$ this requirement states that
    the changes between $\pi_\mathit{actual}$ and $\pi_\mathit{close}$ are a subset of the changes between $\pi_\mathit{actual}$ and $\pi_\mathit{far}$, where we define the changes between two traces $\pi_0,\pi_1$ as $\mathit{changes}(\pi_0,\pi_1) = \{(a,i) \ | \ \pi_0[i] \neq_{\{a\}} \pi_1[i] \}$. For example, let $\pi = \{x\}(\{\})^\omega, \pi_0 = \{\}(\{\})^\omega, \pi_1 = \{\}\{y\}(\{\})^\omega$ 
    and $I = \{x,y\}$. Then, 
    $\pi_0  \leq_\pi \pi_1$, since $\mathit{changes}(\pi,\pi_0) = \{(x,0)\} \subseteq \{(x,0),(y,1)\} = \mathit{changes}(\pi,\pi_1)$. The trace $\pi_2 = \{x\}(\{y\})^\omega$, however, is incomparable to $\pi_0$ and $\pi_1$, as $\mathit{changes}(\pi,\pi_2) = \{(y,j) \ | \ j \geq 1 \}$ is not in any subset relationship with the respective sets for $\pi_0, \pi_1$.
\end{example}

The similarity relations considered in previous works~\cite{BeutnerFFS23,CoenenFFHMS22} are all fundamentally based on $\leq^{\mathit{subset}}$ as defined in Example~\ref{ex:subsetdis}, with added conditions to avoid infinite chains of closer traces. This is directly tied to the \emph{limit assumption} first studied by Lewis in his seminal work on counterfactual modal logic~\cite{Lewis73a}. In our setting, this assumption can be formalized as follows.

\begin{definition}[Limit Assumption]
    A similarity relation $\leq \ \subseteq (2^I)^\omega \times (2^I)^\omega \times (2^I)^\omega$ satisfies the \emph{limit assumption}, if for all traces $\pi \in (2^{I\cup O})^\omega$ and all possible causes $\Cause \subseteq (2^I)^\omega$, we have that $(\overline{\Cause},<_\pi)$ is well-founded, i.e., there is no infinite descending chain $\pi_0 >_\pi \pi_1 >_\pi \ldots \text{ with } \pi_i \in \overline{\Cause}$.
    
\end{definition}

This requirement means that there always exist \emph{closest} counterfactual traces that do not satisfy the cause no matter which actual trace we pick (except if all traces satisfy the cause). These closest traces would be ideal candidates for causal analysis, but unfortunately, they do not always exist, in particular not for the similarity relation $\leq^{\mathit{subset}}$, as stated in Proposition~\ref{prop:subset}. Its proof, like all others, is in Appendix~\ref{app:proofs} due to space reasons.

\begin{proposition}\label{prop:subset}
    $\leq^{\mathit{subset}}$ does not satisfy the limit assumption.
\end{proposition}

Since the original definition of Coenen et al.~\cite{CoenenFFHMS22} quantifies universally over closest traces, it can be vacuously satisfied if the similarity relation does not satisfy the limit assumption.
Previous works have therefore added additional constraints. For instance, Beutner et al.~\cite{BeutnerFFS23} propose $\leq^{\mathit{full}}$, which \emph{additionally} to the constraints of $\leq^{\mathit{subset}}$ (cf. Example~\ref{ex:subsetdis}) requires the following:
\begin{align*}
    \mathit{zip}(\pi_{\mathit{actual}},\pi_{\mathit{close}},\pi_{\mathit{far}}) \models \bigwedge_{i\in I} \big(\G \F ( i_\mathit{actual} \not\leftrightarrow  i_\mathit{close}) \rightarrow \G ( i_\mathit{close} \leftrightarrow  i_\mathit{far})\big) \enspace .
\end{align*}
This encodes that whenever $\pi_\mathit{close}$
differs differs from $\pi_\mathit{actual}$ on some input at infinitely many locations, then $\pi_\mathit{far}$ agrees with $\pi_\mathit{close}$ on this input. Hence, on any chain in $<^\mathit{full}_\pi$, infinite changes on some $i \in I$ eventually get converted into finite ones, which ensures finiteness of the chain since there are only finitely many atomic propositions. We confirm that this results in $\leq^{\mathit{full}}$ satisfying the limit assumption. 

\begin{proposition}\label{prop:full}
    $\leq^{\mathit{full}}$ satisfies the limit assumption.
\end{proposition}

While satisfying the limit assumption is, in principle, useful, in the case of $\leq^\mathit{full}$ this comes at a significant cost: Its logical description contains a large conjunction over the inputs, each containing an implication between temporal formulas. Hence, any algorithmic approach to cause synthesis (and checking) that uses $\leq^\mathit{full}$ will scale poorly in the size of $I$. This motivates us to develop a modified definition of temporal causality that can directly work with the smaller, canonical similarity relation $\leq^\mathit{subset}$, while retaining most of the original semantics of Coenen et al.\ for similarity relations that satisfy the limit assumption, such as $\leq^\mathit{full}$.

\subsection{A General Definition of Temporal Causality}\label{subsec:definition}

We now develop our generalized definition of temporal causality for similarity relations that do not satisfy the limit assumption.

The idea behind our generalization stems from counterfactual modal logic as formalized by Lewis~\cite{Lewis73a}. Lewis' semantics a priori only work for total similarity relations, making them unsuitable for our setting. However, they were recently extended to non-total similarity relations by Finkbeiner and Siber~\cite{FinkbeinerS23}. We apply these semantics to our concrete problem to obtain a well-defined notion of causality for similarity relations that do not satisfy the limit assumption. In Section~\ref{subsec:generalization}, we show that our definition retains the original semantics proposed by Coenen et al.\ for similarity relations that satisfy the limit assumption.

\begin{definition}[Temporal Causality]\label{def:gencausality}
    Let $\mathcal{T}$ be a system, $\pi\in \mathit{traces}(\mathcal{T})$ a trace, $\leq_\pi$ a similarity relation, and $\Effect \subseteq (2^\AP)^\omega$ an effect property. We say that $\Cause  \subseteq (2^{I})^\omega$ is a \emph{cause} of $\Effect$ on $\pi$ in $\mathcal{T}$ if the following conditions hold.
	\begin{description}
		\item[SAT:] For all $\pi_0 \in \mathit{traces}(\mathcal{T}) $ such that $\pi_0 =_I \pi$ we have $\pi_0|_I \in \Cause$ and $\pi_0 \in \Effect$.
		\item[CF:] For all $\pi_0 \in \overline{\Cause}$ there is an at-least-as close trace $\pi_1 \in \overline{\Cause}$, i.e., with $\pi_1 \leq_\pi \pi_0$, such that there is a $\pi_2 \in \mathit{traces}(\mathcal{T}) $ with $\pi_1 =_I \pi_2$ and $\pi_2 \in \overline{\Effect}$.
		\item[MIN:] There is no $\Cause' \subset \Cause$ such that $\Cause'$ satisfies \emph{\textbf{SAT}} and \emph{\textbf{CF}}.
	\end{description}	
\end{definition}

 The main idea of the counterfactual criterion CF is that for every trace $\pi_0$ that does not satisfy the cause, there exists a closer trace $\pi_2$ that does not satisfy the cause \emph{and} the effect. The additional quantification over $\pi_1$ is a technicality included because the cause $\Cause \subseteq (2^{I})^\omega$ consists of input sequences while $\pi_2 \in \mathit{traces}(\mathcal{T})$ is a full system trace. It also closely mirrors the structure of Coenen et al.'s PC2 criterion (cf.\ Definition~\ref{def:coenen}) which it neatly generalizes to similarity relations that do not satisfy the limit assumption: If the assumption holds, then a $\pi_2$ is, in particular, required for the \emph{closest} traces $\pi_0$ in $\overline{\Cause}$, for which $\pi_2$ can only be instantiated by themselves. Hence, the closest traces are required to not satisfy the effect (we develop this comparison more formally in Section~\ref{subsec:generalization}). If the limit assumption does not hold and there exists an infinite chain of ever-closer traces $\pi_0 \in \overline{\Cause}$, the condition requires that for all these $\pi_0$ there is a closer $\pi_2$ that avoids the effect, even in infinity: No matter how far we descend on this chain, we are always guaranteed that we can descend further towards a closer counterfactual trace that does not satisfy the effect.

 \begin{example}
     To illustrate these conditions with a concrete example, consider the system from Figure~\ref{fig:example1}, the trace $\pi =\{x,e\}^\omega$, the effect $\Effect = \Lang(\LTLglobally\LTLeventually e)$, and the cause $\Cause = \Lang(\LTLglobally\LTLeventually x)$, with similarity relation $\leq^\mathit{subset}$. It is easy to that SAT is satisfied, as the system is deterministic and $\pi|_I = \{x\}^\omega \in \Cause$ and $\pi \in \Effect$. There is, as discussed in Section~\ref{subsec:limit}, an infinite chain in $(\Lang(\LTLeventually \LTLglobally \lnot x), \leq_{\pi}^\mathit{subset})$ and, hence, no closest trace. We require for all $\pi_0 \in \Lang(\LTLeventually \LTLglobally \lnot x) = \overline{\Cause}$ a $\pi_1 \in \overline{\Cause}$ with $\pi_1 \in \Lang(\LTLeventually \LTLglobally \lnot e) = \overline{\Effect}$ and a $\pi_2 =_I \pi_1$ such that $\pi_2 \in \overline{\Effect}$. In this case, we can pick $\pi_1$ as $\pi_0$ and $\pi_2$ as the corresponding system trace, hence CF is satisfied. To see that MIN is satisfied, consider any strict subset $\Cause' \subset \Cause$. Hence, there is some $\pi' \in \overline{\Cause'}$ such that $\pi' \models \LTLglobally \LTLeventually x$. Then, all system traces $\pi_2$ with $\pi_2 \leq_\pi \pi'$ satisfy $\pi_2 \models \LTLglobally\LTLeventually x$ by the definition of $\leq^\mathit{subset}_\pi$, and in this system this also means $\pi_2 \models \LTLglobally\LTLeventually e$. Hence, $\Cause$ satisfies MIN because no strict subset satisfies CF.
 \end{example}

\begin{remark}
 Note that Definition~\ref{def:gencausality} is not restricted to similarity relations that can be expressed via zipped traces and LTL formulas as used in the previous examples, but instead applies to any comparative similarity relation as defined at the start of this section.
\end{remark}

\subsection{Proving Generalization}\label{subsec:generalization}

This section is dedicated to proving that our generalization (Definition~\ref{def:gencausality}) is conservative, i.e., agrees with Coenen et al.'s original definition whenever the underlying similarity relation satisfies the limit assumption and the actual trace is deterministic. First, we recall Coenen et al.'s definition. 

\begin{definition}[Coenen et al.~\cite{CoenenFFHMS22}]\label{def:coenen}
	Let $\mathcal{T}$ be a system, $\pi\in \mathit{traces}(\mathcal{T})$ a trace, 
    $\leq_\pi$ a similarity relation, and $\Effect \subseteq (2^{O})^\omega$ an effect property. $\Cause \subseteq (2^{I})^\omega$ is a cause of $\Effect$ on $\pi$ in $\mathcal{T}$ if the following three conditions hold.
	\begin{description}
        \item[PC1:] $\pi|_I \in \Cause$ and $\pi \in \Effect$.
		\item[PC2:] For all closest counterfactual traces $\pi_0 \in \overline{\Cause}$, i.e., traces for which there are no closer traces $\pi_1 \in \overline{\Cause}$ with $\pi_1 <_{\pi} \pi_0$, there exists a $\pi_2 \in \mathit{traces}(\mathcal{T})$ such that $\pi_0 =_I \pi_2$ and $\pi_2 \in \overline{\Effect}$.
        \item[PC3:] There is no $\Cause' \subset \Cause$ such that $\Cause'$ satisfies \emph{\textbf{PC1}} and \emph{\textbf{PC2}}.
	\end{description}	
\end{definition}

Unlike in our updated definition, PC1 only works if the actual trace $\pi$ is deterministic. 
If the $\pi$ is nondeterministic, 
the effect can be avoided with no modifications at all to $\pi$ (which is minimal), hence the cause should be empty. PC1 does not reflect this and allows to build a cause that includes $\pi|_I$ (and possibly more), wrongfully implying that a modification of the sequence is required to avoid the effect. 
PC2  may be vacuously satisfied if the similarity relation does not satisfy the limit assumption, as outlined in Section~\ref{subsec:limit}.

\begin{remark}\label{rem:cfautomaton}
    Note that Coenen et al. consider traces $\pi \in \mathit{traces}(\mathcal{C}^\mathcal{T}_\pi)$ of the \emph{counterfactual automaton} $\mathcal{C}^\mathcal{T}_\pi$ for PC2. This automaton models \emph{contingencies}, which allow to partially reset outputs back to as they were on the actual trace $\pi$, and to change the system state accordingly. 
    For PC2 in Definition~\ref{def:coenen}, this means that the closest counterfactual traces $\pi_2$ do not have to avoid the effect themselves, but together with some contingency. This mechanism, inspired by Halpern's modified version of actual causality~\cite{Halpern15}, 
    was extended by
    Coenen et al.~\cite{CoenenDFFHHMS22,CoenenFFHMS22} to lasso-shaped traces and finite state machines to sometimes obtain more accurate causes. However, to guarantee meaningful results, the original system has to have unique output labels. Beutner et al.'s implementation~\cite{BeutnerFFS23} therefore allows to toggle the usage of contingencies. Similarly, our generalization works both with contingencies and without. For the latter case, one simply supplants $\mathcal{T}$ with $\mathcal{C}^\mathcal{T}_\pi$ in both definitions. Our cause synthesis algorithm can also handle contingencies, and our implementation allows to toggle them as a feature. Our theoretical contribution is independent of this detail. For completeness, we provide the definition of $\mathcal{C}^\mathcal{T}_\pi$ as defined by Coenen et al.~\cite{CoenenFFHMS22} in Appendix~\ref{app:cfautomaton}.
\end{remark}

We now proceed to show the equivalence between our definition (Definition~\ref{def:gencausality}) and Coenen et al.'s definition (Definition~\ref{def:coenen}) in case the limit assumption is fulfilled and the actual trace is deterministic. We start with proving the equivalence of the counterfactual conditions CF and PC2, which holds regardless of nondeterminism on the actual trace.

\begin{lemma}\label{lem:cf}
    Let $\mathcal{T}$ be a system, $\pi\in \mathit{traces}(\mathcal{T})$ a trace, $\Cause \subseteq (2^{I})^\omega$ a cause property, and $\Effect \subseteq (2^\AP)^\omega$ an effect property. Let $\leq$ be a similarity relation that satisfies the limit assumption. Then we have that \emph{PC2} is satisfied iff \emph{CF} is satisfied.
\end{lemma}

The full proof of Lemma~\ref{lem:cf} and of Theorem~\ref{thm:equiv} below, appear in Appendix~\ref{app:proofs}. With Lemma~\ref{lem:cf} at hand, we only need to address the differences between PC1 and SAT. It is easy to see that their equivalence fails when behavior on the actual trace $\pi$ is nondeterministic, i.e., when there is another trace that is input-equivalent to $\pi$ but does not satisfy the effect. In such a case, PC1 is satisfied but SAT is not. Hence, our definition is equivalent to Coenen et al.'s definition only in deterministic systems, as we deliberately diverge in the case of nondeterminism on the actual trace. Notably, Lemma~\ref{lem:cf} holds for both deterministic and nondeterministic systems, and determinism is only relevant on the actual trace.
The restriction to output-only effects $\Effect \subseteq (2^O)^\omega$ is inherited from Coenen et al.'s definition, but technically not necessary.
\begin{theorem}\label{thm:equiv}
    Let $\leq$ be a similarity relation that satisfies the limit assumption. Then $\Cause \subseteq (2^{I})^\omega$ is a cause for $\Effect \subseteq (2^O)^\omega$ on a trace $\pi$ that is deterministic in $\mathcal{T}$ according to our definition (Definition~\ref{def:gencausality}) if and only if it is a cause according to Coenen et al.'s definition (Definition~\ref{def:coenen}). 
\end{theorem}

\section{Cause Synthesis}\label{sec:synthesis}

In this section, we develop our algorithm for synthesizing causes. In Section~\ref{subsec:characterization} we formalize the characterization of a cause as the complement of the upper closure of the negated effect, which we have discussed intuitively in Section~\ref{subsec:topology}. In Section~\ref{subsec:integralautomaton} we provide an algorithm for cause synthesis in the $\omega$-regular setting, when the effect is given as a nondeterministic Büchi automaton and the actual trace is in a lasso shape.

\subsection{Proving Our Characterization }\label{subsec:characterization}

For this section, we fix a system $\mathcal{T}$, an actual trace $\pi \in traces(\mathcal{T})$, a similarity relation $\leq$, and an effect $\Effect$. We now show that, if it exists, the cause for $\Effect$ on $\pi$ is the complement of the upward closure of $\overline{\Effect}$ in $(\traces(\mathcal{T}),\leq_\pi)$. Formally, we construct a set $\mathsf{D}$ that is a cause for $\Effect$ on $\pi$ via its complement:
\begin{align*}
    \overline{\mathsf{D}} &= \{\,\rho  \in (2^I)^\omega \,  \mid \, \exists \sigma \in \mathit{traces}(\mathcal{T}) . \ \sigma \leq_\pi \rho \land \sigma \in \overline{\Effect} \,\} \enspace , \text{ hence}\\
    \mathsf{D} &= \{\,\rho \in (2^I)^\omega \,  \mid \, \forall \sigma \in \mathit{traces}(\mathcal{T}) . \ \sigma \leq_\pi \rho \rightarrow \sigma \in \Effect \,\} \enspace .
\end{align*}
The set $\mathsf{D}$ directly corresponds to the (unique) cause if there exists one, and is empty if there is none. We establish this in a series of lemmas, see Appendix~\ref{app:proofs} for their full proofs. 

\begin{lemma}\label{lem:iscause}
If the set $\mathsf{D}$ is non-empty, it is a cause for $\Effect$ on $\pi$ in $\mathcal{T}$.
\end{lemma}

The proof of Lemma~\ref{lem:iscause} shows that $\mathsf{D}$ satisfies Definition~\ref{def:gencausality} assuming it is non-empty. The assumption is only required for SAT, as this criterion requires that $\pi$ and all input-equivalent traces are in the cause. CF follows from the definition of $\mathsf{D}$, and for MIN we can show that any strict subset of $\mathsf{D}$ does not satisfy CF.

\begin{lemma}\label{lem:nocause}
Iff the set $\mathsf{D}$ is empty, there exists no cause that satisfies SAT.
\end{lemma}

Lemma~\ref{lem:nocause} serves two purposes. First, it helps us argue for the completeness of our construction. Second, it shows that the only reason why there may be no cause is due to a nondeterministic actual trace. To fully argue completeness, we show that causes are unique, and hence $\mathsf{D}$ is the only relevant cause in all cases.

\begin{lemma}~\label{prop:uniqueness}
    Causes are (semantically) unique: There can be no two sets $\Cause \neq \Cause'$ that are both causes for some effect property $\Effect$ on a trace $\pi$ in some system $\mathcal{T}$.
\end{lemma}

\begin{remark}
This does \emph{not} mean that there can only exist a single causal event, such as ``$a$ at position 0" or ``$b$ at position 1", in a given scenario. Instead, Lemma~\ref{prop:uniqueness} states that the semantics of the symbolic description of the causal behavior in a given scenario is unique. It is precisely the idea of temporal causality to encompass multiple single events in a single symbolic description, e.g., through a conjunction such as $a \land \LTLnext b$.
\end{remark}

\subsection{Cause-Synthesis Algorithm for $\omega$-Regular Effects}\label{subsec:integralautomaton}

In Section~\ref{subsec:characterization}, we have established a direct characterization of causes as downward closed sets, independent of any concrete descriptions of cause, effect, and trace. In this section, we develop an automata-based algorithm for synthesizing causes of $\omega$-regular effects given, e.g., by a nondeterministic Büchi automaton~(NBA), on lasso-shaped traces. We assume that the relation $\leq \ \subseteq (2^I)^\omega \times (2^I)^\omega \times (2^I)^\omega$ is definable by a relational $\omega$-regular property $\prop_\leq \subseteq (2^{I \times \{t_0,t_1,t_2\}})^\omega$, such that
$(\pi_0,\pi_1,\pi_2) \in \ \leq$ iff the zipped trace $\mathit{zip}(\pi_0,\pi_1,\pi_2)$ satisfies $\prop_\leq$. Note that this applies to all concrete similarity relations introduced in Section~\ref{sec:definitions}.  We show that under these assumptions, the set $\mathsf{D}$ from Section~\ref{subsec:characterization} can be constructed as an NBA. First, we construct an NBA for $\overline{\mathsf{D}}$ and subsequently complement it. This is necessary because we start out from an NBA representation for the effect, and assume the similarity relation to be given by an NBA as well. Since the NBAs acceptance condition is existential, we need the additional complementations to express the universal quantification over the closer traces $\sigma$ appearing in the definition of $\mathsf{D}$. 

The main technical difficulty that remains is to ensure that the conditions on the three traces $\pi_\mathit{actual}$, $\pi_\mathit{close}$ and $\pi_\mathit{far}$, as they appear in the alphabet of a similarity relation, are applied consistently, and that the quantification over $\sigma$ in $\mathsf{D}$, which corresponds to $\pi_\mathit{close}$, is resolved at the correct step, as the automaton should range over the inputs $I$ and not, e.g., over $I \times \{t_0,t_1,t_2\}$ as used by the similarity relation. 

\paragraph{Similarity Relation.} Our starting point is the NBA for the similarity relation defined by the $\omega$-regular property $\mathsf{P}_\leq$: $\mathcal{A}_{\leq}^I = (Q_{\leq},2^{I \times \{t_0,t_1,t_2\}}, Q_{\leq}^0, F_{\leq}, \Delta_{\leq}^I)$. The automaton $\mathcal{A}_{\leq}^I$ only reasons about inputs and uses tuples with the trace variables $t_0,t_1$ and $t_2$ to encode whether the input appears on the actual, closer or farther trace, respectively. We lift the automaton to the full set of atomic propositions as the automaton $\mathcal{A}_{\leq} = (Q_{\leq},2^{(I \times \{t_0,t_1,t_2\}) \cup (O \times \{t_0,t_1\})}, Q_{\leq}^0, F_{\leq}, \Delta_{\leq})$. The transition relation is defined as follows, for a letter $w$: $q_2 \in \Delta_{\leq}(q_1,w) \text{ iff } q_2 \in \Delta_{\leq}^I\big(q_1,w \setminus (O \times \{t_0,t_1\})\big)$.
Hence, $\mathcal{A}_{\leq}$ specifies the same relation between the inputs of the three traces as $\mathcal{A}_{\leq}^I$, but allows arbitrary output behavior. Its alphabet does not contain outputs for $\pi_2$, as these traces eventually form the elements of the cause, which only ranges over the inputs.

\paragraph{Effect.} Next, we modify the NBA $\mathcal{A}_{\Effect}^* = (Q_{\Effect},2^\AP, q_{\Effect}, F_{\Effect}, \Delta_{\Effect}^*)$ for the $\omega$-regular effect $\Effect$ such that it refers to the closer trace $t_1$ and ranges over the same alphabet as $\mathcal{A}_{\leq}$. We obtain $\mathcal{A}_{\Effect} = (Q_{\Effect},2^{(I \times \{t_0,t_1,t_2\}) \cup (O \times \{t_0,t_1\})}, q_{\Effect}, $ $F_{\Effect}, \Delta_{\Effect})$ with:
\begin{align*}
&q_2 \in \Delta_{\Effect}\big(q_1,(w \times \{t_1\}) \cup X \cup Y\big)  \text{ iff} \\
q_2 \in \Delta_{\Effect}^*&(q_1,w) \land X \subseteq (\AP \times \{t_0\}) \land Y \subseteq (I \times \{t_2\})\enspace .
\end{align*}
Hence, $\mathcal{A}_{\Effect}$ restricts $\pi_1$ to be in $\Effect$ by restricting it to the transition relation of $\mathcal{A}_{\Effect}^*$, while allowing an arbitrary trace $\pi_0$ and arbitrary input sequence in $\pi_2$.

\paragraph{Intersection.} For the conjunction that defines the set $\overline{\mathsf{D}}$, we intersect $\mathcal{A}_{\leq}$ with the complement of $\mathcal{A}_{\Effect}$ to obtain $\mathcal{A}_{\cap} = (Q_{\cap},2^{(I \times \{t_0,t_1,t_2\}) \cup (O \times \{t_0,t_1\})}, Q_{\cap}^0, F_{\cap}, \Delta_{\cap})$ such that: $\mathcal{A}_{\cap} = \mathcal{A}_{\leq} \cap \overline{\mathcal{A}_{\Effect}}$.

\paragraph{System Product.}  As the next step, we construct the product of the automaton $\mathcal{A}_{\cap}$ with the system $\mathcal{T} = (S,s_0,\AP,\delta,l)$, ensuring that the atomic propositions $t_1$ are picked from a valid system trace. When building the product, we project away explicit atomic propositions paired with $t_1$, as the traces of the desired set $\mathsf{D}$ are only the traces paired with $t_2$.  The resulting automaton is $\mathcal{A}_{\times} = (S \times Q_{\cap},$ $ 2^{( I \times \{t_0,t_2\}) \cup (O \times \{t_0\})}, \{s_0\} \times Q_{\cap}^0, S \times F_{\cap}, \Delta_\times)$, where
\begin{align*}
\Delta_\times \big((s_i,q_i),w\big) = \big\{(s_{i+1},q_{i+1})\mid \ &\exists A \subseteq I . \ s_{i+1} \in \delta(s_i,A) \, \land \\
&q_{i+1} \in \Delta_{\cap}\big(q_i,((l(s_{i+1}) \cup A) \times \{t_1\})\cup w\big)\big\} \enspace .
\end{align*}

\paragraph{Cause Automaton.} To obtain the final result, we first complement the automaton from the previous step to obtain $\overline{\mathcal{A}_{\times}} = (Q_{\times},2^{( I \times \{t_0,t_2\}) \cup (O \times \{t_0\})}, Q_{\times}^0, F_{\times}, \Delta_\times)$, and then build the product with the trace. At the same step we project away atomic propositions paired with $t_0$, and remove the trace variable $t_2$ to obtain the alphabet $2^I$ for the cause.
For the lasso-shaped trace $\pi  = \pi_0 \ldots \pi_{j-1} \cdot (\pi_{j} \ldots \pi_k)^\omega$ we define the set of positions as $\Pi = \{0,\ldots,k\}$ and a successor function $\mathit{succ}: \Pi \mapsto \Pi$ as $\mathit{succ}(r) = r+1$ for $r < k$, and $\mathit{succ}(k) = j$. The cause automaton is then $\mathcal{A}_{\mathsf{D}} = (\Pi \times Q_{\times},2^I, \{\pi_0\} \times Q_{\times}^0, \Pi \times F_{\times}, \Delta_\mathsf{D})$, where
\begin{align*}
\Delta_\mathsf{D} \big((i,q_i),w\big) = \big\{(\mathit{succ}(i),q_{i+1})\mid \  q_{i+1} \in \Delta_{\times}\big(q_i,(\pi_i \times \{t_0\}) \cup (w \times \{t_2\})\big)\big\} \enspace .
\end{align*}

From the lemmas established in Section~\ref{subsec:characterization}, we conclude that there is a cause iff $\mathcal{A}_{\mathsf{D}}$ is non-empty, and then the cause is uniquely determined by its language. 

\begin{corollary}
    The language of $\mathcal{A}_\mathsf{D}$ is empty iff there is no cause $\Cause$ for $\Effect$ on $\pi$ in $\mathcal{T}$, and if $\Lang(\mathcal{A}_\mathsf{D})$ is non-empty, then it is the unique cause for $\Effect$ on $\pi$ in $\mathcal{T}$.
\end{corollary}

We can also state an upper bound on the size of $\mathcal{A}_{\mathsf{D}}$, which is dominated by the potentially exponential growth from NBA complementation~\cite{Schewe09}. 

\begin{proposition}
    If the effect $\mathsf{E}$ and the similarity relation $\leq$ are given as NBAs $\mathcal{A}_\mathsf{E}$ and $\mathcal{A}_{\leq}$, respectively, then the size of $\mathcal{A}_\mathsf{D}$ is in $|\pi| \cdot 2^{\widetilde{\mathcal{O}}(2^{\widetilde{\mathcal{O}}(|\mathcal{A}_{\Effect}|)} \cdot |\mathcal{A}_{\leq}| \cdot |\mathcal{T}|)}$.
\end{proposition}

Note that the doubly-exponential upper bound in the description of $\Effect$ persists independent of whether it is given as an NBA or LTL formula. In the latter case, we simply translate the negated formula, which again leads to an exponential blow-up. In theory, the description does make a difference for $\leq$: If it is represented as a formula, we first need to translate it with a potentially exponential increase in size, hence it would move up one exponent in the bound. In practice, the canonical similarity relation $\leq^\mathit{subset}$ can always be represented by a 1-state NBA, such that its contribution to the bound is less relevant.

While the stated upper bound may seem daunting, it mirrors the (tight) bounds of related problems, such as LTL synthesis~\cite{PnueliR89}. In the following section, we show that, not only can our approach solve many cause-synthesis problems in practice, it also significantly improves upon previous methods for cause checking.

\section{Implementation and Evaluation}\label{sec:experiments}

In this section, we evaluate a prototype tool implementing our cause-synthesis approach, called CORP - \emph{Causes for Omega-Regular Properties}.\footnote{Our prototype is on GitHub: \url{https://github.com/reactive-systems/corp}. Our full evaluation can be reproduced with the artifact on Zenodo: \url{https://doi.org/10.5281/zenodo.10946309}.} Our prototype is written in Python and uses Spot~\cite{Duret-LutzRCRAS22} for automata operations and manipulation. The prototype allows for both cause synthesis and cause checking, where in the latter case the correct cause is first synthesized and than checked for equivalence against the cause candidate. This allows for a direct comparison with the cause checking tool CATS~\cite{BeutnerFFS23} in Section~\ref{sec:checkingexperiments}. Before, we report on our experiments on cause synthesis, where we compare our method with the incomplete, sketch-based approach of CATS. All experiments were carried out on a machine equipped with a 2.8 GHz Intel Xeon processor and 64GB of memory, running Ubuntu 22.04. 

\subsection{Cause Synthesis}\label{sec:synthesisexperiments}
We conducted three different experiments that highlight how the similarity relations, effect size and system size contribute to the performance of our algorithm.

\begin{table}[!t]
	\caption{Cause synthesis for arbiters. $|\mathcal{T}|$ is the number of system states.  In all instances, $\pi$ is the (unique) trace where all $n$ clients send a request at every position, which has length $|\pi| = n$. $\varphi_\Effect$ is the effect. We report the time taken to synthesize the causal NBA with the metrics $\leq^\mathit{full}$ and $\leq^\mathit{subset}$ in seconds and the respective NBA sizes $|\mathcal{A}_\Cause^\mathit{full}|$ and $|\mathcal{A}_\Cause^\mathit{subset}|$, and provide an LTL description $\varphi_\Cause$ of the NBA language (guessed manually). TO denotes the timeout of $60$ seconds.}\label{tab:arbiter}
	   \vspace{1.25em}
		\centering
		\def\arraystretch{1.1}
		\setlength\tabcolsep{3mm}
		\begin{booktabs}{lcccccccc}
			\toprule
			\textbf{Instance} & $\boldsymbol{|\mathcal{T}|}$ & $\boldsymbol{\varphi_\Effect}$  & $\boldsymbol{t(\leq^\mathit{full})}$  & $\boldsymbol{t(\leq^\mathit{subset})}$ & $\boldsymbol{|\mathcal{A}_\Cause^\mathit{full}|}$  & $\boldsymbol{|\mathcal{A}_\Cause^\mathit{subset}|}$ & $\boldsymbol{\varphi_\Cause}$ \\
			\midrule
			\textsc{Spurious 1} & 1 & $\LTLeventually g_0$ & 0.11 & 0.11 & 1 & 1 & $\true$\\
			\midrule[dotted]
			\textsc{Spurious 2} & 2 & $\LTLeventually g_0$& 0.11 & 0.11 & 1 & 1 & $\true$\\
			\midrule[dotted]
			\textsc{Spurious 3} & 3 & $\LTLeventually g_0$ & 0.21 & 0.11 & 1 & 1 & $\true$\\
			\midrule[dotted]
			\textsc{Spurious 4} & 4 & $\LTLeventually g_0$ & TO & 0.11 & - & 1 & $\true$\\
                \midrule
			\textsc{Unfair 2} & 2 & $\LTLglobally \lnot g_0$& 0.11 & 0.11 & 1 & 1 & $\LTLglobally r_\mathit{prio}$\\
			\midrule[dotted]
			\textsc{Unfair 3} & 4 & $\LTLglobally \lnot g_0$ & 0.16 & 0.11 & 1 & 1 & $\LTLglobally r_\mathit{prio}$\\
			\midrule[dotted]
			\textsc{Unfair 4} & 6 & $\LTLglobally \lnot g_0$ & TO & 0.11 & - & 1 & $\LTLglobally r_\mathit{prio}$\\
			\midrule
			\textsc{Full 1} & 1 & \makecell{$\LTLeventually g_0$ \\ $\LTLglobally \LTLeventually g_0$} & \makecell{0.11 \\ 0.11} & \makecell{0.11\\ 0.11} &  \makecell{2 \\ 2}  &  \makecell{2 \\ 2}  & \makecell{$\LTLeventually r_0$ \\ $\LTLglobally \LTLeventually r_0$} \\
			\midrule[dotted]
			\textsc{Full 2} & 4 & \makecell{$\LTLeventually g_0$ \\ $\LTLglobally \LTLeventually g_0$} & \makecell{0.11 \\ 0.11} & \makecell{0.11\\ 0.11} &  \makecell{2 \\ 12}  &  \makecell{2 \\ 4} & \makecell{$\LTLeventually r_0$ \\ $\LTLglobally \LTLeventually r_0$} \\
			\midrule[dotted]
			\textsc{Full 3} & 11 & \makecell{$\LTLeventually g_0$ \\ $\LTLglobally \LTLeventually g_0$} & \makecell{0.16 \\ 2.04} & \makecell{0.11\\ 0.16} &  \makecell{2 \\ 215}  &  \makecell{2 \\ 24}  & \makecell{$\LTLeventually r_0$ \\ $\LTLglobally \LTLeventually r_0$} \\
			\midrule[dotted]
			\textsc{Full 4} & 46 & \makecell{$\LTLeventually g_0$ \\ $\LTLglobally \LTLeventually g_0$} & \makecell{TO \\ TO} & \makecell{0.11\\ 33.22} &  \makecell{- \\ -}  &  \makecell{2 \\ 214} & \makecell{$\LTLeventually r_0$ \\ $\LTLglobally \LTLeventually r_0$} \\   
			\bottomrule
		\end{booktabs}
\end{table}

\subsubsection{Arbiters.} 
We computed causes on traces of resource arbiters to compare the performance of our algorithm under different similarity relations, whose logical description scales in the number of system inputs. An arbiter instance is parameterized by a number of clients $n$, each with its own input. This let us easily scale the size of the similarity relation's description. For some number $n$ of clients (indexed by $k$) that request access to a shared resource with a request $r_k$, an arbiter grants mutually exclusive access to the resource with a grant $g_k$. We considered different arbiter strategies, and for each we synthesize causes as NBAs $\mathcal{A}_\Cause^\mathit{full}$ and $\mathcal{A}_\Cause^\mathit{subset}$ with the similarity relations $\leq^\mathit{full}$ and $\leq^\mathit{subset}$, respectively. The results of these instances are depicted in Table~\ref{tab:arbiter}. Spurious arbiters simply give out grants to all clients in a round-robin manner, regardless of previous requests. Unfair arbiters prioritize one client with request $r_\mathit{prio}$ over the others, while full arbiters are fully functional arbiters that only give out grants that were requested beforehand. In all instances, we computed causes on the (unique) trace $\pi$ where all clients send requests continuously, i.e., $\pi|_I = \{r_0,\ldots,r_n\}^\omega$. Consequently, on this trace both the spurious and the full arbiter send grants to all clients, while the unfair arbiter only gives grants to the prioritized client. These varying strategies are reflected in the synthesized cause-effect pairs. In the spurious arbiters, the language of the synthesized cause NBA for the effect $\LTLeventually g_0$ is $\true$, which reflects that the effect appears on \emph{all} system traces. In the unfair arbiters, the cause for no grant being given to client $0$ is that the prioritized arbiter sends requests permanently, i.e., the causal NBA has the language $\LTLglobally r_\mathit{prio}$. In the full arbiters, $\LTLeventually g_0$ is caused, as expected, by $\LTLeventually r_0$ and $\LTLglobally \LTLeventually g_0$ is caused by $\LTLglobally \LTLeventually r_0$. From a performance standpoint, the arbiter instances show us that accommodating the canonical similarity relation $\leq^\mathit{subset}$, as we did through our generalization of temporal causality in Section~\ref{sec:definitions}, leads to significant improvements in practice: In all instances, synthesizing causes with $\leq^\mathit{subset}$ was faster than with $\leq^\mathit{full}$, and the resulting causal NBAs were smaller as well. This is mostly because of the number of inputs involved: The other parameters stay comparably small when going from the spurious 1-arbiter to the spurious 4-arbiter, but the latter times out when using $\leq^\mathit{full}$. When the systems get larger and the effects more complex, e.g., in the instance of the full 4-arbiter with the effect $\LTLglobally \LTLeventually g_0$, the automata produced can become bigger even with $\leq^\mathit{subset}$. However, the language of the produced automata has a small representation, i.e., $\LTLglobally \LTLeventually r_0$, such that we see potential for improvement through automata minimization techniques.

\begin{wrapfigure}{R}{0.5\linewidth}
        \vspace{-1.8em}
	\centering 
        \includegraphics[width=0.5\textwidth]{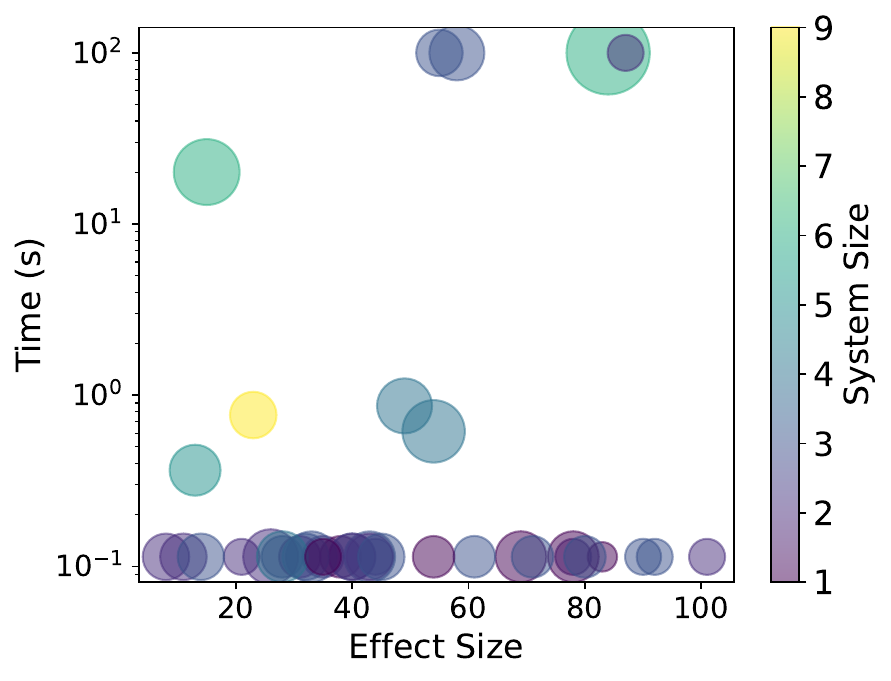}
        \caption{Computing causes for neural synthesis mispredictions with CORP. Size of a point represents the length of the counterexample (between 2 and 16).}\label{fig:neurosynth}
        \vspace{-1.5em}
\end{wrapfigure}

\subsubsection{Neural Synthesis.} For more diverse effects, we considered mispredicted circuits from a neural synthesis model~\cite{SchmittHRF21}. Given some specification (in this case, generated by Spot's \texttt{randltl}) the neural model predicts an implementation as an AIGER~\cite{Biere-FMV-TR-07-1} circuit, which is in the end model-checked against the specification. Since neural synthesis is not sound, this check fails occasionally and returns a counterexample, which may be used for further repair~\cite{CoslerSHF23}. We used our tool CORP to compute the cause for the violation of the specification on such a counterexample. 
In Figure~\ref{fig:neurosynth} we report the time of computing causes with respect to
size of the syntax tree of the effect formula, and the system size. The timeout was set to 100 seconds. The size of the points in the scatter plot corresponds to the length of the counterexample and the color to the system size. From the plot we can deduce that a large effect does not per-se mean a long runtime of our tool. However, a combination of large effects, bigger systems, and longer counterexamples usually means that the tool takes longer. The sizes of the synthesized causes are diverse and range from 2 to 60 states.

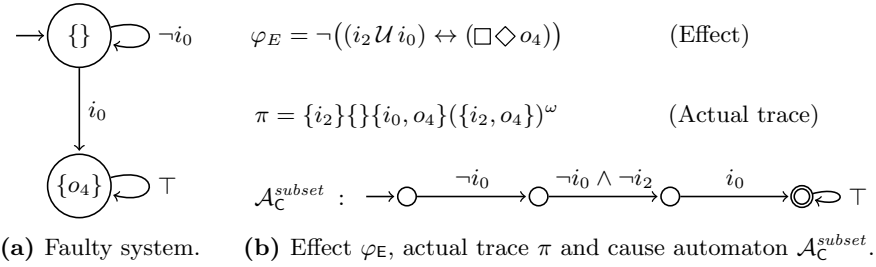
\begin{figure}[t!]
	\begin{subfigure}{0.3\linewidth}
		\centering
  \begin{tikzpicture}[->,shorten >=1pt,semithick,auto,node distance=2cm, on grid,initial text=,
			every state/.style={minimum size=24pt,inner sep=1pt}]
         \node[state,initial left](1){$\{\}$};
         \node[state,below = of 1](2){$\{o_4\}$};
         \path[->,draw](1) edge[loop right] node{$\lnot i_0$} (1)
         (1) edge[] node[pos=0.5]{$i_0$} (2)
         (2) edge[loop right] node{$\top$} (2);
         \end{tikzpicture}
         \caption{Faulty system.}\label{subfig:system}
	\end{subfigure}%
	\begin{subfigure}{0.7\linewidth}
		\centering
  \begin{tikzpicture}[->,shorten >=1pt,semithick,auto,node distance=1.75cm, on grid,initial text=,
			every state/.style={minimum size=7pt,inner sep=1pt}]
   
         \node[state,initial left](1){};
        \node[left = of 1,xshift=1em](trace){$\mathcal{A}_\Cause^\mathit{subset}\,$ :};
         \node[state,right = of 1](2){};
         \node[state,right = of 2](3){};
         \node[state,right = of 3,accepting](4){};
         \path[->,draw](1) edge[] node[pos=0.5]{$\lnot i_0$} (2)
         (2) edge[] node[pos=0.5]{$\lnot i_0 \land \lnot i_2$} (3)
         (3) edge[] node[pos=0.5]{$i_0$} (4)
         (4) edge[loop right] node[pos=0.5]{$\top$} (4);
         \node[above = of 1,yshift=-2.1em](trace){$\pi = \{i_2\}\{\}\{i_0,o_4\}(\{i_2,o_4\})^\omega$};
        \node[right = of trace,xshift=8.4em](trc){(Actual trace)};
         \node[above = of trace,yshift=-2.1em](effect){$\varphi_E = \lnot\big((i_2 \LTLuntil i_0 ) \leftrightarrow ( \LTLglobally \LTLeventually o_4 )\big)$};
        \node[right = of effect,xshift=7em](eff){(Effect)};
         \end{tikzpicture}
         \caption{Effect $\varphi_\Effect$, actual trace $\pi$ and cause automaton $\mathcal{A}_\Cause^\mathit{subset}$.}\label{subfig:rest}
	\end{subfigure}
	\caption{A system predicted wrongly by a neural synthesis model (Figure~\ref{subfig:system}) for the specification $\lnot \varphi_\Effect$, i.e., the negation of the effect. The effect is shown together with the actual trace $\pi$, i.e., a counterexample obtained from model checking, and the computed cause automaton $\mathcal{A}_\Cause^\mathit{subset}$ in Figure~\ref{subfig:rest}. }\label{fig:neuralsynthesis}
\end{figure}

\begin{example}
    We discuss an illustrative example of cause synthesis with a small benchmark from the neural synthesis datatset. All relevant inputs and outputs of our cause synthesis algorithm are depicted in Figure~\ref{fig:neuralsynthesis}. First, we have the system (cf. Figure~\ref{subfig:system}), which is a wrongly predicted circuit of the neural synthesis model. This model tried to come up with a solution for the specification $(i_2 \LTLuntil i_0 ) \leftrightarrow ( \LTLglobally \LTLeventually o_4 )$, i.e., $o_4$ appears infinitely often if and only if input $i_2$ is enabled until input $i_0$ is enabled. The predicted system does not satisfy this specification, because there are cases where $\LTLglobally \LTLeventually o_4$ holds without the inputs meeting the required condition. Hence, model checking the specification returns a counterexample $\pi$ that violates the formula, which means the negated specification can be seen as an effect $\varphi_\Effect$ that is present on the counterexample $\pi$ (cf. Figure~\ref{subfig:rest}). Our algorithm then computes the cause for this effect, i.e., for the violation of the specification, on the counterexample $\pi$, as a nondeterministic Büchi automaton. The computed automaton $\mathcal{A}_\Cause^\mathit{subset}$ is depicted at the bottom of Figure~\ref{subfig:rest}. It is language-equivalent to the LTL formula $\lnot i_0 \land \LTLnext (\lnot i_0 \land \lnot i_2 \land \LTLnext i_0)$, which basically states that the effect is caused by a conjunction of four inputs spread out over the first three steps. Indeed, it is easy to see that modifying any of these four inputs results in a trace that satisfies the specification: For instance, setting $i_0$ at the first position results in the trace that immediately enters the state labeled with $o_4$ and loops there forever such that the left part of the equivalence is satisfied, while removing $i_0$ from the third position results in looping in the initial state such that the right part of the equivalence is not satisfied anymore.
\end{example}

\subsubsection{Comparison with Cause Sketching.} CATS, the tool of Beutner et al.~\cite{BeutnerFFS23}, allows to enumerate non-temporal formulas in holes of a provided cause sketch until a cause is found. If the effect contains $\LTLnext$ as the only temporal operator and a cause exists, there is a sketch that is guaranteed to encompass the cause. This provides us with a baseline with which we can compare our cause-synthesis algorithm. We constructed random benchmarks that fall into CATS' complete fragment using Spot's \texttt{randaut} function to generate systems with $10$ up to $1000$ states, obtaining traces of length $2$ and then inserting a new atomic proposition $e$ at the last position of the trace and in the system. The effect then is defined as the occurrence of $e$ at exactly this position. We chose such small traces and effects because CATS timed out already on slightly larger instances. We conducted additional experiments using just our tool CORP with traces (and effects) of size 10. Figure~\ref{fig:sketching} shows the time taken by CATS and CORP to synthesize causes. The influence of the system size on the runtime of CORP in this setting is negligible, which we believe is due to the efficient automata operations performed by Spot. The hyperproperty encoding of CATS does not seem as amenable to similar optimizations.

\begin{figure}[t!]
	\caption{Direct comparisons between our tool CORP and the tool CATS~\cite{BeutnerFFS23}. Figure~\ref{fig:sketching} shows the time CATS needs to synthesize a cause in its complete fragment with trace and effect of size 2, and the time taken by CORP for sizes 2 and 10. Table~\ref{tab:checking} shows the time taken to check single causal relationships. These problems are taken from Beutner et al.~\cite{BeutnerFFS23} (where ``Instances" are ``Examples").}\label{tab:causechecking}
	\begin{subfigure}{0.5\linewidth}
		\centering
            \caption{Cause synthesis.}\label{fig:sketching}
            \vspace{-2pt}
            \includegraphics[width=\textwidth]{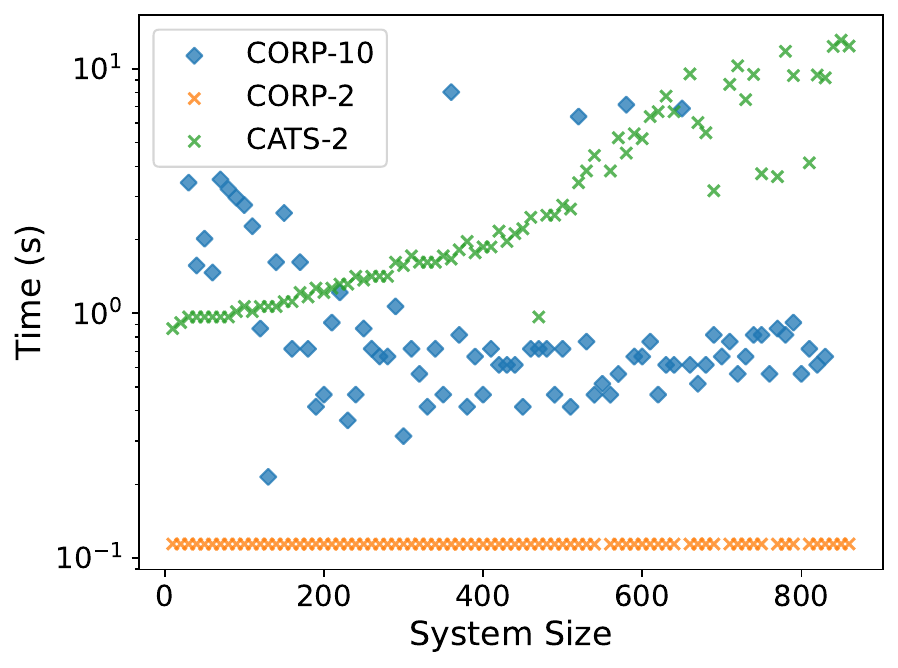}
	\end{subfigure}%
	\begin{subfigure}{0.5\linewidth}
		\centering
            \caption{Cause checking.}\label{tab:checking}
		\def\arraystretch{1.3}
		\setlength\tabcolsep{1mm}
		\setlength\dashlinedash{1pt}
		\setlength\dashlinegap{2pt}
		\setlength\arrayrulewidth{0.7pt}
		\begin{tabular}{lcc}
			\toprule
			\textbf{Instance} & CATS & CORP \\
			\midrule
			\textsc{Spurious Arbiter}  & 0.72 & 0.11 \\
			\textsc{Simple Arbiter} & 0.97 & 0.11 \\
			\textsc{Arbiter} & 22.84 & 0.11 \\
			\textsc{Instance 6, odd} & 1.81 & 0.11 \\
			\textsc{Instance 8} & 0.92 & 0.11 \\
			\textsc{TP Right} & 0.87 & 0.11 \\
			\bottomrule
		\end{tabular}
	\end{subfigure}
\end{figure}

\subsection{Cause Checking}\label{sec:checkingexperiments}

It is straightforward to use our cause synthesis algorithm to also check causes through an equivalence check between the synthesized causal NBA and the candidate formula (or automaton). This allows a direct performance comparison with the cause checking tool CATS of Beutner et al.~\cite{BeutnerFFS23}, which we conducted on the publicly available benchmarks of their paper. In these cause-checking benchmarks, a cause candidate is given in addition to the system, actual trace and effect. The time CATS and our tool CORP took in each instance to check whether the given candidate is a cause is depicted in Figure~\ref{tab:checking}. Somewhat surprisingly, our cause checker based on cause synthesis performs significantly better on all benchmarks. This shows that our characterization of causes as complements of the upward closure of the negated effect (cf. Section~\ref{subsec:characterization}) is more efficient than encoding the cause-checking instances into a hyperlogic, as done by CATS.

\section{Related Work}
The study of causality and its applications in formal methods has gained great interest in recent years~\cite{BaierDFJMPZ21}. In a finite setting, Ibrahim et al.\ use SAT solvers and linear programming to check~\cite{IbrahimRP2019} and infer~\cite{IbrahimP20} actual causes. Our definition of actual causality for reactive systems extends the definitions of Coenen et al.~\cite{CoenenFFHMS22} to cases in which the limit assumption does not hold. While Coenen et al.\ study the theory of actual causality~\cite{Halpern15} in reactive systems, they do not provide a way to \emph{generate} causes and explanations. 
In terms of cause \emph{synthesis}, the most related work is by Beutner et al.~\cite{BeutnerFFS23}, which checks causality and generates causes based on sketching. Unlike ours, their tool is only applicable for the small fragment of LTL containing only $\LTLnext$ operators, while we are able to generate temporal causes for all $\omega$-regular specifications.  

In a series of works, Leue et al.\  
study symbolic description of counterfactual causes in \emph{Event Order Logic}~\cite{CaltaisGL18,Leitner-FischerL13,Leitner-FischerL14}. 
However, this logic can only reason about the ordering of events, and not their absolute timing, as we can do with $\omega$-regular properties (e.g., specifying that the input at the second position is the cause). 

Gössler and Métayer~\cite{GosslerM13} define causality for component-based systems, and Gössler and Stefani~\cite{GosslerS20} study causality based on counterfactual builders. Their formalisms differ from ours, which is based on Coenen et al.~\cite{CoenenFFHMS22}, and none of the works considers cause \emph{synthesis}.

 Most other works related to cause synthesis concern generating explanations for effects observed on {finite} traces~\cite{BallNR03,GroceCKS06,Groce04,WangYIG06}, or effects restricted to safety properties~\cite{ParreauxPB23}.
 In the context of cause synthesis over infinite traces for effects given as temporal specifications, 
 existing works are limited to causes given as sets of events (i.e., atomic propositions and times points)~\cite{BeerBCOT09,CoenenDFFHHMS22,HorakCMHFMDFD22} or take a state-centric view to, e.g., measure the responsibility of a state for an observed effect property~\cite{DBLP:conf/aaai/BaierBK0P24,BaierDFJPZ22,DBLP:conf/lics/MascleBFJK21}.

\section{Conclusion}
This paper presents the first complete algorithm to compute temporal causes for arbitrary $\omega$-regular properties. It is based on a new, generalized version of temporal causality that solves a central dilemma of previous definitions by loosening the assumptions on similarity relations. From a philosophical perspective, this is an immense step forward since it is the first definition that accommodates the canonical similarity relation used in previous literature. Our experimental results show that our generalization also leads to significant improvements from a practical perspective. These mainly stem from characterizing causes based on set-closure properties, which may be an interesting approach for counterfactual causality in other formalisms. Besides, our work opens up exciting research directions on generating explanations from temporal causes, i.e., as formulas or annotations in highlighted counterexamples. 

\subsubsection{Acknowledgements.}

We thank Matthias Cosler and Frederik Schmitt for providing us with the neural synthesis benchmarks. This work was partially supported by the DFG in project 389792660 (Center for Perspicuous Systems, TRR 248) and by the ERC Grant HYPER (No. 101055412).

\bibliographystyle{splncs04}
\bibliography{bibliography}

\appendix
\section{Appendix}

\subsection{Additional Definitions}\label{app:cfautomaton}

\begin{definition}[Counterfactual Automaton~\cite{CoenenFFHMS22}]\label{def:CFautomaton}
Let $\mathcal{T} = (S, s_0, \AP, \delta, l)$ be a system and $\pi = \pi_0 \ldots \pi_i \cdot (\pi_{j} \ldots \pi_{k})^\omega \in \traces(\mathcal{T})$ be a lasso-shaped trace. The \emph{counterfactual automaton} for $\pi$ and $\mathcal{T}$ is $\mathcal{C}^\mathcal{T}_\pi = (S^\mathcal{C}, s^\mathcal{C}_0, I^\mathcal{C} \cup O, \delta^\mathcal{C}, l^\mathcal{C})$, where:
\begin{itemize}
    \item $S^\mathcal{C} = S \times \{0 \ldots k\}$, we have $k$ copies of the original system;
    \item $s^\mathcal{C}_0 = (s_0,0)$, paths start in the initial state of the first copy;
    \item $I^\mathcal{C} = I \cupdot \{o^{\mathcal{C}}~|~o \in O \}$, additional inputs for setting an output as contingency;
    \item $(s',n') \in \delta^\mathcal{C}((s,n),Y)$ iff the following holds:
    \begin{enumerate}
        \item\label{con:1} if $n = k$ then  $n'=j$ else $n' = n+1$, and
        \item\label{con:2} there is some $s'' \in \delta(s,Y|_I)$ such that for all $o^{\mathcal{C}} \in Y$: $o \in l(s') \leftrightarrow o \in \pi_n$ and for all $o^{\mathcal{C}} \not\in Y$: $o \in l(s') \leftrightarrow o \in l(s'')$;
    \end{enumerate}
    \item $l^\mathcal{C}((s,k)) = l(s)$, the labeling function is based on the original states.
\end{itemize}
\end{definition}

\subsection{Omitted Proofs}\label{app:proofs}

\setcounter{proposition}{0} 
\setcounter{theorem}{0} 
\setcounter{lemma}{0} 

\begin{proposition}
    $\leq^{\mathit{subset}}$ does not satisfy the limit assumption.
\end{proposition}

\begin{proof}
    We generalize from the example of Coenen et al.~\cite{CoenenFFHMS22}. Let $I = \{a\}$ and consider the actual trace $\pi = \{a\}^\omega$ and $\Cause = \LTLglobally \LTLeventually a$. Let $\pi_1$ be any trace in $\overline{\Cause} = \LTLeventually \LTLglobally \lnot a$, consequently there exists some $i$ such that for all $j \geq i : a \not\in \pi_1[j]$. Then we can construct a closer trace $\pi_2$ with $a \in \pi_2[i]$ and $\pi_2[k] = \pi_1[k]$ for all $k \neq i$. Thus we have such a $\pi_2 \in \overline{\Cause}$ with $\pi_2 <_\pi \pi_1$ for any $\pi_1 \in \overline{\Cause}$, which means we cannot pick a $\pi_1$ for any $\pi_0 \in \overline{\Cause}$ as required by the limit assumption.\qed
\end{proof}

\begin{proposition}
    $\leq^{\mathit{full}}$ satisfies the limit assumption.
\end{proposition}

\begin{proof}
    Assume there were $\Cause$ and $\pi$, with an infinite descending chain $\pi_0 >^{\mathit{full}}_\pi \pi_1 >^{\mathit{full}}_\pi \ldots$ with $\pi_i \in \overline{\Cause}$ for all $i$ . Since $>^{\mathit{full}}_\pi$ is based on $>^{\mathit{subset}}_\pi$ and is irreflexive, we have that $\mathit{changes}(\pi_i,\pi) \supset \mathit{changes}(\pi_{i+1},\pi)$ for all $i$. Given $A \subseteq \AP$ and $X \subseteq \AP \times \mathbb{N}$, we define $X|_A = X \cap \{(a,i) \ | \ a \in A \land i \in \mathbb{N} \}$. For every $i$, let $F_i \subseteq I$ be the largest set such that $\mathit{changes}(\pi_i,\pi)|_{F_i}$ is finite, i.e., $F_i$ is the sets of inputs
on which $\pi_i$ and $\pi$ differ on finitely many positions. 
We have that 
$F_i \subseteq F_{i+1}$ for all $i$, as the other case contradicts $\mathit{changes}(\pi_i,\pi) \supset \mathit{changes}(\pi_{i+1},\pi)$. Since $I$ is finite, there can only be finitely many $i$ with $F_i \subset F_{i+1}$. Let $j = i+1$ for the largest such $i$, hence we have $F_k = F_{k+1}$ for all $k \geq j$. From the additional criterion of $<^\mathit{full}_\pi$ it follows that $\mathit{changes}(\pi_k,\pi)|_{I \setminus F_k} = \mathit{changes}(\pi_{k+1},\pi)|_{I \setminus F_{k+1}}$, hence we must have $\mathit{changes}(\pi_k,\pi)|_{F_k} \supset \mathit{changes}(\pi_{k+1},\pi)|_{F_{k+1}}$. Since $\mathit{changes}(\pi_k,\pi)|_{F_k}$ is finite by construction, we can have only a finite number of such $k$'s, which yields a contradiction with our starting chain being infinite.\qed
    
\end{proof}

\begin{lemma}
    Let $\mathcal{T}$ be a system, $\pi\in \mathit{traces}(\mathcal{T})$ a trace, $\Cause \subseteq (2^{I})^\omega$ a cause property, and $\Effect \subseteq (2^\AP)^\omega$ an effect property. Let $\leq$ be a similarity relation that satisfies the limit assumption. Then we have that \emph{PC2} is satisfied iff \emph{CF} is satisfied.
\end{lemma}

\begin{proof} Since PC2 and CF first quantify universally over $\overline{\Cause}$, it is easy to see that they are both satisfied in the case where $\Cause = (2^{I})^\omega$ and $\overline{\Cause}$ is, consequently, empty. We therefore in the following assume that $\overline{\Cause}$ is non-empty. Since $\leq$ satisfies the limit assumption, we have a non-empty set of closest counterfactual traces $C_\pi = \{\rho \in \overline{\Cause} \ | \ \forall \sigma \in \overline{\Cause}. \ \sigma \not<_\pi \rho \}$, i.e., all minimal elements in $(\overline{\Cause},<_\pi)$. We write $E(\pi_1)$ if $\pi_1$ has an input-equivalent trace avoiding the effect, i.e., $E(\pi_1)$ iff $ \exists \pi_2 \in \mathit{traces}(\mathcal{T}) . \ \pi_2 =_I \pi_1 \land \pi_2 \in \overline{\Effect}$. We now prove the directions of the equivalence separately.\\
``$\Rightarrow$" : Assume that PC2 is satisfied. Let $\pi_0 \in \overline{\Cause}$ be any trace not satisfying the cause. We know there is some $\pi_1 \in C_\pi$ with $\pi_1 \leq_\pi \pi_0$,
as either 
$\pi_0$ is a minimal element in $(\overline{\Cause},<_\pi)$, or a closer trace is. 
From PC2 it then follows that $E(\pi_1)$, implying CF.\\
``$\Leftarrow$" : Assume that CF is satisfied. Since $C_\pi \subseteq \overline{\Cause}$, for all $\pi_0 \in C_\pi$ there is a $\pi_1 \in \overline{\Cause}$ with $\pi_1 \leq_\pi \pi_0$ such that $E(\pi_1)$. Since $\pi_0$ is in $C_\pi$, it holds that $\pi_1 = \pi_0$. Hence, we have $E(\pi_0)$ for all $\pi_0 \in C_\pi$.\qed
\end{proof}

\begin{remark}
As in Definitions~\ref{def:gencausality} and~\ref{def:coenen} (cf. Remark~\ref{rem:cfautomaton}), the system $\mathcal{T}$ in the proof of Lemma~\ref{lem:cf} could easily be replaced by $\mathcal{C}^\mathcal{T}_\pi$ to search for a contingency trace that does not satisfy the effect, i.e., in the definition of $E(\rho)$. Hence, Lemma~\ref{lem:cf} and Theorem~\ref{thm:equiv} which we develop based upon it, hold regardless of whether contingencies are applied in the respective definitions.
\end{remark}

\begin{theorem}
    Let $\leq$ be a similarity relation that satisfies the limit assumption. Then $\Cause \subseteq (2^{I})^\omega$ is a cause for $\Effect \subseteq (2^O)^\omega$ on a trace $\pi$ that is deterministic in $\mathcal{T}$ according to our definition (Definition~\ref{def:gencausality}) if and only if it is a cause according to Coenen et al.'s definition (Definition~\ref{def:coenen}). 
\end{theorem}

\begin{proof}
    With Lemma~\ref{lem:cf} at hand, we only need to show that SAT is satisfied in a scenario iff PC1 is satisfied, because MIN and PC3 state the same requirement.\\
    ``$\Rightarrow$" : This direction is trivial, as $\pi =_I \pi$.\\
    ``$\Leftarrow$" : Since we assume $\pi$ is deterministic in $\mathcal{T}$, we have for all $\pi_0 \in \traces(\mathcal{T})$ that $\pi_0 =_I \pi$ implies $\pi_0 = \pi$. From PC1 we know $\pi_0|_I \in \Cause$ and $\pi \in \Effect$. Hence, $\pi_0|_I \in \Cause$ and $\pi_0 \in \Effect$.\qed
\end{proof}

\begin{lemma}
If the set $\mathsf{D}$ is non-empty, it is a cause for $\Effect$ on $\pi$ in $\mathcal{T}$.
\end{lemma}

\begin{proof} Assume $\mathsf{D}$ is non-empty. We show that it satisfies SAT, CF, and MIN of Definition~\ref{def:gencausality}.

SAT: Consider $\pi' \in \mathsf{D}$. 
For all traces $\pi_0$ that are input equivalent to the actual trace, i.e., $\pi_0 =_I \pi$, we have that $\pi_0 \leq_\pi \pi'$ by definition.
Therefore, $\pi_0 \in \Effect$ for all such $\pi_0$. Any $\pi_1 \leq_\pi \pi_0$ must have $\pi_1 =_I \pi$ as $\pi$ is the unique minimum of $\leq_\pi$. It follows $\pi_1 \in \Effect$ for all such $\pi_1$, which means $\pi_0|_I \in \mathsf{D}$ for all $\pi_0 =_I \pi$. Hence, $\mathsf{D}$ satisfies SAT.

CF: If $\overline{\mathsf{D}}$ is empty, CF is trivially satisfied. So consider some $\pi_0 \in \overline{\mathsf{D}}$. By the definition of $\mathsf{D}$ it follows that there is some $\sigma \in \mathit{traces}(\mathcal{T})$ with $\sigma \leq_\pi \pi_0$ and $\sigma \notin \Effect$. Then we can pick $\pi_1 = \sigma$ for the CF condition with $\pi_2 = \pi_1$, which implies $\pi_2 =_I \pi_1$ and $\pi_2 \notin \Effect$. This shows that $\mathsf{D}$ satisfies CF.

MIN: For contradiction, assume there was a $\mathsf{D}' \subset \mathsf{D}$ that satisfies SAT and CF. Then there is some $\pi_0 \in  \overline{\mathsf{D}'}$ such that $\pi_0 \in \mathsf{D} \setminus \mathsf{D}'$. Since $\pi_0 \in \mathsf{D}$, for all $\pi_1 \in \mathit{traces}(\mathcal{T})$ we have $\pi_1 \leq_\pi \pi_0 \rightarrow \pi_1 \in \Effect$. But since $\pi_0 \in \overline{\mathsf{D}'}$ and $\mathsf{D}'$ satisfies CF, we have some $\pi_1 \leq_\pi \pi_0$ and $\pi_2 =_I \pi_1$, which together implies $\pi_2 \leq_\pi \pi_0$, with $\pi_2 \notin \Effect$, a contradiction. \qed

\end{proof}

\begin{lemma}
Iff the set $\mathsf{D}$ is empty, there exists no cause that satisfies SAT.
\end{lemma}

\begin{proof}
     ``$\Rightarrow$" : We show that if $\mathsf{D}$ is empty, SAT can never be satisfied. In this case, we in particular have $\pi|_I \notin \mathsf{D}$, which means there is some $\sigma \leq_\pi \pi$ with $\sigma \notin \Effect$. Since $\pi$ is the unique minimum of $\leq_\pi$, we must have $\sigma =_I \pi$. This means we have a spoiling $\pi_0 = \sigma$ with $\pi_0 \notin \Effect$ for all possible candidates $\Cause \subseteq (2^I)^\omega$, so no cause exists.\\
     ``$\Leftarrow$" : If SAT is not satisfied for \emph{any} cause $\Cause$, it is also not satisfied for $\Cause' = \{\pi|_I\}$, i.e., the cause that contains only the input sequence of the actual trace $\pi$. Since for all $\pi_0 =_I \pi$ we have $\pi_0 \in \Cause'$ by construction, there must exists some $\pi_0 =_I \pi$ with $\pi_0 \notin \Effect$. For contradiction, assume that $\mathsf{D}$ was non-empty and there is some $\sigma \in \mathsf{D}$ such that all $\sigma_0 \leq_\pi \sigma$ satisfy $\sigma_0 \in \Effect$. Then, we have in particular $\pi \leq_\pi \sigma$ by the definition of a similarity relation, and with $\pi_0 =_I \pi$ it follows that $\pi_0 \leq_\pi \sigma$. This yields a contradiction: since $\pi_0 \notin \Effect$ but all $\sigma_0 \leq_\pi \sigma$ satisfy $\sigma_0 \in \Effect$.  \\
     Note that this scenario directly corresponds to nondeterminism on the actual trace.\qed
\end{proof}

\begin{lemma}
    Causes are (semantically) unique: There can be no two sets $\Cause \neq \Cause'$ that are both causes for some effect property $\Effect$ on a trace $\pi$ in some system $\mathcal{T}$.
\end{lemma}

\begin{proof}
    By contradiction. Assume there were two distinct causes $\Cause \neq \Cause'$ for some effect property $\Effect$ on a trace $\pi$ in some system $\mathcal{T}$. We show that the intersection $\Cause \cap \Cause'$ also satisfies SAT and CF, which contradicts $\Cause$ and  $\Cause'$ satisfying MIN.
    
    SAT: From $\Cause$ satisfying SAT we know for all $\pi_0 =_I \pi$ that $\pi_0 \in \Cause$ and $\pi_0 \in \Effect$, from $\Cause'$ satisfying SAT we further know $\pi_0 \in \Cause'$, hence we have $\pi_0 \in \Cause \cap \Cause'$, such that $\Cause \cap \Cause'$ satisfies SAT.

    CF: Consider some $\pi_0 \in \overline{\Cause \cap \Cause'}$, so either $\pi_0 \in \overline{\Cause}$ or $\pi_0 \in \overline{\Cause'}$. If $\pi_0 \in \overline{\Cause}$, we can pick $\pi_1$ and $\pi_2$ as in the CF condition of $\Cause$, because $\pi_1 \in \overline{\Cause}$ implies $\pi_1 \in \overline{\Cause \cap \Cause'}$. In the other case, we can pick $\pi_1$ and $\pi_2$ symmetrically as in the CF condition of $\Cause'$ to again obtain a $\pi_1 \in \overline{\Cause \cap \Cause'}$. In both cases, we have $\pi_1 =_I \pi_2$ and $\pi_2 \notin \Effect$, which shows that $\Cause \cap \Cause'$ satisfies CF. \qed
\end{proof}

\setcounter{proposition}{2}

\begin{proposition}
    If the effect $\mathsf{E}$ and the similarity relation $\leq$ are given as NBAs $\mathcal{A}_\mathsf{E}$ and $\mathcal{A}_{\leq}$, respectively, then the size of $\mathcal{A}_\mathsf{D}$ is in $|\pi| \cdot 2^{\widetilde{\mathcal{O}}(2^{\widetilde{\mathcal{O}}(|\mathcal{A}_{\Effect}|)} \cdot |\mathcal{A}_{\leq}| \cdot |\mathcal{T}|)}$.
\end{proposition}

\begin{proof}
    By construction, the size of $\mathcal{A}_\times$ is in $2^{\widetilde{\mathcal{O}}(|\mathcal{A}_\Effect|)} \cdot |\mathcal{A_\leq}| \cdot |\mathcal{T}|$, since NBA complementation of an automaton of size $n$ leads to at most $2^{\widetilde{\mathcal{O}}(n)}$ many states~\cite{Schewe09}. With the same argument, we can deduce that the size of the resulting automaton $\overline{\mathcal{A}_\times}$ is in $2^{\widetilde{\mathcal{O}}(2^{\widetilde{\mathcal{O}}(|\mathcal{A}_{\Effect}|)} \cdot |\mathcal{A}_{\leq}| \cdot |\mathcal{T}|)}$, which is equivalent to the claim.\qed
\end{proof}

\end{document}

%% file: macros.tex
\usepackage{amsmath,amssymb}
\usepackage{tikz}
\usepackage{hyperref}

\usepackage{subcaption} 
\usepackage[strings]{underscore}	 
\usepackage{pifont}
\usepackage{cleveref}
\usepackage{todonotes}
\usepackage{xspace}
\usepackage{tikz}
\usepackage{ltl}
\usetikzlibrary{arrows,decorations.pathmorphing,automata,positioning,hobby,patterns,patterns.meta}
\usepackage{wrapfig}

\usetikzlibrary{calc}

\makeatletter
\def\moverlay{\mathpalette\mov@rlay}
\def\mov@rlay#1#2{\leavevmode\vtop{%
		\baselineskip\z@skip \lineskiplimit-\maxdimen
		\ialign{\hfil$\m@th#1##$\hfil\cr#2\crcr}}}
\newcommand{\charfusion}[3][\mathord]{
	#1{\ifx#1\mathop\vphantom{#2}\fi
		\mathpalette\mov@rlay{#2\cr#3}
	}
	\ifx#1\mathop\expandafter\displaylimits\fi}
\makeatother

\newcommand{\cupdot}{\charfusion[\mathbin]{\cup}{\cdot}}

\newcommand{\AP}{\mathit{AP}}
\newcommand{\Vap}{\mathit{V}}
\newcommand{\Wap}{\mathit{W}}

\newcommand{\Cause}{\mathsf{C}}
\newcommand{\Effect}{\mathsf{E}}

\newcommand{\In}{\mathit{I}}
\newcommand{\Out}{\mathit{O}}

\newcommand{\true}[0]{\mathit{true}}
\newcommand{\false}[0]{\mathit{false}}

\newcommand{\ldot}{\mathpunct{.}}

\newcommand{\Lang}{\mathcal{L}}

\newcommand{\U}{\LTLuntil}
\newcommand{\X}{\LTLnext}
\newcommand{\G}{\LTLglobally}
\newcommand{\F}{\LTLeventually}
\newcommand{\R}{\LTLrelease}

\newcommand{\traces}{\mathit{traces}}

\newcommand{\prop}{\mathsf{P}}

\renewcommand{\models}{\vDash}
\newcommand{\nmodels}{\nvDash}

\newcommand{\donotshow}[1]{}
